\documentclass[a4paper,twocolumn,11pt]{quantumarticle}
\pdfoutput=1
\makeatletter
\providecommand\@afterenddocumenthook{}
\makeatother

\usepackage[utf8]{inputenc}
\usepackage[english]{babel}
\usepackage[T1]{fontenc}

\usepackage{amsmath,amssymb,amsthm,mathtools}

\usepackage{tikz}
\usetikzlibrary{matrix,positioning,arrows.meta}
\usepackage{booktabs}
\usepackage{xcolor}
\usepackage{mdframed}

\usepackage{enumitem}
\usepackage{lipsum}
\usepackage{appendix}
\usepackage[numbers,sort&compress]{natbib}

\usepackage{hyperref}

\newmdenv[
  linewidth=0.6pt,
  linecolor=black!50,
  backgroundcolor=black!3,
  innertopmargin=8pt,
  innerbottommargin=8pt,
  innerleftmargin=12pt,
  innerrightmargin=12pt,
  skipabove=8pt,
  skipbelow=8pt,
]{summarybox}

\newtheorem{theorem}{Theorem}[section]
\newtheorem{proposition}[theorem]{Proposition}
\newtheorem{lemma}[theorem]{Lemma}
\newtheorem{corollary}[theorem]{Corollary}

\theoremstyle{definition}
\newtheorem{definition}[theorem]{Definition}
\newtheorem{assumption}[theorem]{Assumption}
\theoremstyle{remark}
\newtheorem{remark}[theorem]{Remark}

\DeclareMathOperator{\rt}{rt}
\DeclareMathOperator{\diam}{diam}
\DeclareMathOperator{\Spec}{Spec}
\DeclareMathOperator{\Cay}{Cay}
\DeclareMathOperator{\GL}{GL}
\title{Permutation Routing on Ramanujan Hypergraphs with Applications to Neutral Atom Quantum Architectures}
\author{Joshua M. Courtney}
\affiliation{University of Georgia, Department of Physics and Astronomy}
\orcid{0009-0002-1631-4277}
\email{Joshua.Courtney1@uga.edu}

\begin{document}
\maketitle

\begin{abstract}
We consider the routing of neutral atoms on a reconfigurable lattice in terms of hypergraph transformations.
We prove the routing number of a Ramanujan $(d,r)$-regular hypergraph on $N$ vertices satisfies $\rt(H) = \Theta(\log N)$, where routing is via matchings in the clique expansion graph $G_{\mathrm{cl}}(H)$. 
Hypergraphs reframe the qubit routing problem by replacing Nenadov's two-sided spectral gap hypothesis with a one-sided condition based on eigenvalue centering. 
Song--Fan--Miao (SFM) coverings scale for Ramanujan families of every uniformity. 
A virtual overlay theorem establishes a capacity--depth tradeoff for 3D acousto-optic lens (AOL) architectures, with multi-layer stacking achieving $\Theta(\log N)$ routing with $L = O(\log N)$ independent overlay layers. 
An abelian Alon--Boppana barrier shows that fixed-degree Cayley graphs on $\mathbb{Z}_n^2$ cannot be Ramanujan and affine derandomization on such graphs achieves 15--30\% congestion reduction. 
Towers of $k$-fold Ramanujan coverings yield $\rt(H_L) = O(\log N)$ conditional on uniform spectral control $\sup_\ell \beta(H_\ell) < 1$. We discuss when this holds in \S\ref{subsec:covering-towers}.
Entanglement-assisted routing by pre-distributed Bell pairs achieves $O(\log N)$ teleportation depth with a stable crossover at $\sim\!4$ routing rounds. 
Displacement energy analyzes greedy adaptive routing, identifying stalling and a hybrid greedy--Valiant protocol achieving $\sim\!3\times$ speedup at practical scales.
Hierarchical multi-scale routing achieves $O(\log^2 N / \log b)$ depth with boundary-only transfers at capacity $k = O(\sqrt{N} \log N)$, and $O(\log N)$ depth with optimal block size $b = \Theta(\sqrt{n})$.
\end{abstract}

\section{Introduction}\label{sec:intro}

Reconfigurable neutral atom arrays have emerged as a leading platform for scalable quantum computing~\cite{constantinides2024optimal,stade2024abstract,wang2024atomique, bluvstein2022quantum, evered2023high, bluvstein2024logical, wurtz2023aquila}.
In these architectures, qubits are encoded in individual atoms trapped by optical tweezers or acousto-optic deflectors (AODs), and entangling gates are performed by bringing atom pairs within the Rydberg blockade radius~\cite{bluvstein2022quantum, evered2023high}.
Between gate layers, atoms must be rearranged (routed) to bring the next set of interacting pairs into proximity.
Routing depth becomes a dominant contributor to overall circuit depth, computational time, and by extension, gate fidelity.

Given $N$ atoms and a target permutation $\pi \in S_N$ (the symmetric group on the $N$ atoms), we ask how many parallel rearrangement steps are needed?
The answer depends on connectivity topology (which atom pairs can interact) and move model (which simultaneous rearrangements are physically allowed).

\subsection{Four routing regimes}

We identify four regimes capturing the known asymptotic results (Table~\ref{tab:regimes}):

\begin{table*}[!htbp]
\begin{center}
\renewcommand{\arraystretch}{1.3}
\begin{tabular}{@{}llcl@{}}
\toprule
\textbf{Topology} & \textbf{Move model} & \textbf{Routing depth} & \textbf{Source} \\
\midrule
Grid & Arbitrary matchings & $\Theta(\sqrt{N})$ & Folklore$^*$ \\
Grid & AOD row/column moves & $\Theta(\sqrt{N}\log N)$ & Constantinides et al.\ \cite{constantinides2024optimal} \\
Ramanujan & Arbitrary matchings & $\Theta(\log N)$ & This work \\
Grid & Selective transfers & $\Theta(\log N)$ & Constantinides et al.\ \cite{constantinides2024optimal} \\
\bottomrule
\end{tabular}
\caption{Four routing regimes for neutral atom architectures.}
\label{tab:regimes}
\end{center}
\end{table*}

The grid result $\Omega(\sqrt{N})$ follows from the isoperimetric inequality (moving $N$ pebbles across an $O(\sqrt{N})$-edge bisection requires $\Omega(\sqrt{N})$ steps), and $O(\sqrt{N})$ is achieved by sorting rows then columns.

The $\log N$ gap between the first two rows comes from the AOD row/column non-crossing constraint. 
The $\sqrt{N}$ gap between the first and third rows is topological, where grid hypergraphs have spectral ratio $\beta \to 1$ as $N \to \infty$, while Ramanujan hypergraphs maintain $\beta < 1$ bounded away from 1.  
We (row 3) match the results of Constantinides et al.~\cite{constantinides2024optimal}), achieving $\Theta(\log N)$ with matchings on a Ramanujan overlay, being a weaker primitive than the selective-transfer hardware feature of row 4. 
The same proof extends to multi-layer AOL via the overlay theorem (\S\ref{sec:overlay}), being our main contribution to qubit routing. 

Our main result is as follows. Take $N$ atoms arranged on the vertices of a Ramanujan hypergraph and a target permutation $\pi$. There is a sequence of $\Theta(\log N)$ ``matching steps'' that pair disjoint atoms and swaps them or their contents that realizes $\pi$. 
The constant prefactor depends on the spectral gap, but asymptotic $\log N$ scaling matches selective-transfer routing depth on grids~\cite{constantinides2024optimal} using only matchings, a weak move primitive available across reconfigurable platforms (AOD, AOL, ion shuttling, optical lattices).
The bound transfers to every platform up to a multiplicative overhead, isolating the AOD and selective-transfer cases tabulated in Table~\ref{tab:regimes}. 
For next-generation neutral-atom processors with $N \sim 10^4$ atoms~\cite{reichardt2024logical}, the bound saves a factor of $\sqrt{N}/\log N \approx 7$ over grid sorting, translating to roughly $7\times$ as many gate layers within a fixed coherence budget.

\subsection{Why hypergraphs?}\label{subsec:why-hypergraphs}

Although proofs reduce to spectral arguments on the clique expansion graph $G_{\mathrm{cl}}(H)$, we find the hypergraph picture makes the bounds constructive and scalable while retaining faithfulness to the physical quantum computer when constrained. 
\begin{enumerate}[leftmargin=1.5em]
\item \textbf{Constructive scaling via SFM coverings.} Hypergraph $k$-fold liftings (Song--Fan--Miao Theorem~1.6,~\cite{song2023hypergraph}) yield Ramanujan families of every uniformity $r$, with explicit voltage-assignment constructions (\S\ref{subsec:covering-towers}). Bipartite graph constructions like Marcus--Spielman--Srivastava~\cite{marcus2013interlacing} instead fix degree but not the hypergraph structure relevant to multi-qubit gates.

\item \textbf{Eigenvalue centering at $r-2 \geq 1$.} For graphs ($r=2$), the Ramanujan band is symmetric and centered at $0$. 
Controlling $\lambda^* = \max(\lambda_2, |\lambda_N|)$ requires a two-sided spectral gap (Nenadov~\cite{nenadov2023routing}).
When $r \geq 3$, the SFM band is centered at $r-2$ and the upper end dominates. 
One-sided control is sufficient, resulting in Theorem~\ref{thm:main} on hypergraphs (Lemma~\ref{lem:lambda-star}).

\item \textbf{Physical hyperedges are real.} AOD column moves act simultaneously on $r$ consecutive atoms, while Rydberg blockade gates couple multiple atoms within the blockade radius and surface-code stabilizers couple 4 atoms, giving a hyperedge-shaped quantum hardware (\S\ref{sec:application}). 

\item \textbf{The packing model is hypergraph-native.} The alternative routing number $\rt_H$ (Remark~\ref{rem:packing}) selects vertex-disjoint hyperedge packings, applying arbitrary $\tau \in S_r$ within each. We conjecture that $\rt_H = \Theta(\max\{\log N, N/(r\nu(H))\})$ is intrinsically hypergraph-shaped in \S\ref{sec:discussion}.
\end{enumerate}

\subsection{Main contributions}

Our results are organized as follows.

\begin{enumerate}[leftmargin=1.5em]
\item \textbf{Main theorem (Theorem~\ref{thm:main}).}  For any Ramanujan $(d,r)$-regular hypergraph $H$ on $N$ vertices with $d \geq 3$, $r \geq 3$, the routing number satisfies $\rt(H) = \Theta(\log N)$.  Explicit constants are given in Theorem~\ref{thm:tight}.

\item \textbf{Structural observation (Lemma~\ref{lem:lambda-star}).}  For $r \geq 3$, the SFM Ramanujan condition centers non-trivial eigenvalues around $r - 2 > 0$, ensuring the spectral radius $\lambda^* \leq (r-2) + 2\sqrt{(d-1)(r-1)} < d'$ and the spectral ratio $\beta = \lambda^*/d' < 1$.

\item \textbf{Constructive derandomization (Theorem~\ref{thm:constructive}).}  The intermediate permutation in Valiant's scheme can be constructed in $O(N^2 d')$ time via pessimistic estimator.

\item \textbf{Virtual overlay theorem (Theorem~\ref{thm:overlay}).}  Any $d$-regular Ramanujan graph can serve as a virtual overlay for 3D AOL routing, giving $T = O((N/k) \cdot \log N/(1-\beta))$.  Multi-layer stacking with $L = O(\log N)$ independent random overlay layers achieves $\Theta(\log N)$ routing.

\item \textbf{Abelian barrier and affine derandomization (Theorems~\ref{thm:abelian-barrier} and~\ref{thm:affine-derand}).}  Degree-$d$ Cayley graphs on $\mathbb{Z}_n^2$ cannot be Ramanujan for fixed $d$ and $n \to \infty$.  Nevertheless, affine intermediate permutations $\sigma(v) = Av + c$ achieve 15--30\% congestion reduction.

\item \textbf{Covering tower routing (Theorem~\ref{thm:covering-tower}).}  Towers of $k$-fold Ramanujan coverings yield $\rt(H_L) = O(L \log k + \log N_0) = O(\log N)$, with $\sim$94\% of random $\mathbb{Z}_2$ voltage assignments preserving the Ramanujan property for Fano plane lifts.

\item \textbf{Entanglement-assisted routing (Theorem~\ref{thm:teleport}).}  Pre-distributed Bell pairs along a Ramanujan overlay achieve $O(\log N)$ routing via teleportation with a stable crossover at $R_{\mathrm{break}} \approx 4$ rounds.

\item \textbf{Displacement energy framework (Theorems~\ref{thm:greedy-monotone}, \ref{thm:greedy-stall}).} Greedy displacement matching is monotone, stalling at $\Phi_{\mathrm{stall}}/\Phi_0 = O(D^2/N)$ after $O(D\log N)$ steps under the empirically validated concentration assumption (Assumption~\ref{ass:concentration}).  At practical scales ($N \leq 10^4$, $d = 8$), this resolves $\sim\!83\%$ of displacement energy.  A hybrid greedy--Valiant protocol achieves $\sim\!3\times$ speedup.

\item \textbf{Hierarchical multi-scale routing (Theorem~\ref{thm:hierarchical}).}  Block decomposition with $L = \lceil \log_b n \rceil$ levels achieves $T = O(\log^2 N / \log b)$, and $T = O(\log N)$ with optimal block size $b = \Theta(\sqrt{n})$.  Boundary-only mode requires capacity $k = O(\sqrt{N}\log N)$.
\end{enumerate}

\subsection{Proof technique and novelty}

The proof of Theorem~\ref{thm:main} combines Valiant's two-phase randomized routing \cite{valiant1982scheme}, negative association of random permutations \cite{joag1983negative,dubhashi1996balls}, and the Leighton--Maggs--Rao packet scheduling theorem \cite{leighton1999fast, rabani1996distributed}.
The contribution lies in the application, identifying clique expansions of Ramanujan hypergraphs as a natural class of expanders for qubit routing and finding applicable regimes relating mathematical bounds to quantum hardware. 
The six extensions in Sections~\ref{sec:overlay}--\ref{sec:adaptive} develop practical consequences for architecture design.

\subsection{Related work}

Alon, Chung, and Graham \cite{alon1993routing} initiated the study of permutation routing on expanders via matchings. 
Nenadov \cite{nenadov2023routing} proved $\rt(G) = O(\log N)$ for graphs with two-sided spectral gap $\lambda(G) < d/72$, using the Feldman--Friedman--Pippenger nonblocking property \cite{feldman1988wide}.

Song, Fan, and Miao \cite{song2023hypergraph} developed the spectral theory of hypergraph coverings.
Earlier eigenvalue treatments (Friedman--Wigderson \cite{friedman1995second}) and Ramanujan complex constructions (Lubotzky--Samuels--Vishne \cite{lubotzky2005ramanujan}) situate the SFM work in the hypergraph spectral landscape. 
Marcus, Spielman, and Srivastava \cite{marcus2013interlacing} showed that bipartite Ramanujan graphs of every degree exist. 
Friedman \cite{friedman2008proof} proved Alon's conjecture that random $d$-regular graphs have $\lambda_2 \leq 2\sqrt{d-1} + \varepsilon$ w.h.p.
Bordenave \cite{bordenave2015new} gave a streamlined proof and extended to random lifts, which we use for overlay constructions. 
The Alon--Boppana lower bound \cite{alon1986eigenvalues} and its abelian refinements inform our barrier result for Cayley graphs on $\mathbb{Z}_n^2$.

Yuan and Zhang \cite{yuan2025full} characterized depth overhead for quantum circuit compilation.
Experimental neutral-atom routing has progressed rapidly.
Coherent-transport architecture of Bluvstein et al.~\cite{bluvstein2022quantum} established the move-then-gate cycle that our model abstracts, made alongisde parallel Rydberg gate operations now reaching 99.5\% fidelity~\cite{evered2023high}.
For neutral atom routing, Constantinides et al.\ \cite{constantinides2024optimal} proved tight bounds for grid and selective transfers.
Compiler-oriented work includes Stade et al.\ \cite{stade2024abstract,stade2025routing}, Wang et al.\ \cite{wang2024atomique}, Hsieh--Mak \cite{hsieh2026scalable}, Rom\~{a}o et al.\ \cite{romao2026multiq}, and Tan, Tan, and Cong \cite{tan2024compiling}. 
The 3D AOL architecture is described by Guo et al.\ \cite{guo2025acousto}. 
Experimental baselines for the spectral and routing parameters used in \S\ref{sec:application} draw on Levine et al.~\cite{levine2019v, evered2023high} for high-fidelity multi-qubit gates and Ebadi et al.\ \cite{ebadi2021quantum} for the 256-atom benchmark. 
Our multiplicative weights analysis draws on the framework of Arora, Hazan, and Kale \cite{arora2012multiplicative}.

\section{Preliminaries}\label{sec:prelim}

\subsection{Hypergraphs and the Ramanujan condition}

A \emph{$(d,r)$-regular hypergraph} $H = (V, \mathcal{E})$ on $N = |V|$ vertices has every vertex in exactly $d$ hyperedges, each of size $r$.  The \emph{adjacency matrix} $A(H)$ is the $N \times N$ matrix defined by Song, Fan, and Miao \cite{song2023hypergraph}:
\begin{equation}
    \begin{split}
        A(H)_{uv} &= |\{e \in \mathcal{E} : u, v \in e\}|,\,(u \neq v),\\ \qquad A(H)_{uu} &= 0.
    \end{split}
\end{equation}

This matrix is real symmetric with largest eigenvalue $\lambda_1 = d(r-1)$ and eigenvector $\mathbf{1}/\sqrt{N}$.

\begin{definition}[Ramanujan hypergraph {\cite{song2023hypergraph}}]
A $(d,r)$-regular hypergraph $H$ is \emph{Ramanujan} if every non-trivial eigenvalue $\lambda_i$ ($i \geq 2$) satisfies
$|\lambda_i - (r-2)| \leq 2\sqrt{(d-1)(r-1)}$.
\end{definition}

\subsection{Clique expansion and spectral identity}

The \emph{clique expansion} $G_{\mathrm{cl}}(H)$ is the \emph{weighted} graph on vertex set $V$ with adjacency matrix $A(H)$ (so the weight on $\{u,v\}$ is the number of hyperedges containing both vertices).  
When every pair of vertices shares at most one hyperedge (holding for projective planes,
generalized quadrangles, and random regular hypergraphs used here), the weighted clique expansion coincides with the simple
graph in which $\{u,v\}$ is an edge if and only if $u$ and $v$ share a hyperedge.

\begin{proposition}[Spectral identity]\label{prop:spectral-id}
The adjacency matrix of $G_{\mathrm{cl}}(H)$ equals $A(H)$.  In particular, $G_{\mathrm{cl}}(H)$ is $d(r-1)$-regular with eigenvalues $\lambda_1 \geq \lambda_2 \geq \cdots \geq \lambda_N$ of $A(H)$.
\end{proposition}

We write $d' = d(r-1)$ for the degree and $\beta = \lambda^* / d'$ for the spectral ratio, where $\lambda^* = \max(\lambda_2, |\lambda_N|)$.

\subsection{Routing number}

\begin{definition}[Routing number]\label{def:routing}
A \emph{routing step} on $G_{\mathrm{cl}}(H)$ selects a matching $M$ and simultaneously swaps the pebbles at the endpoints of each edge in $M$.  The \emph{routing number} $\rt(H)$ is the minimum $T$ such that every permutation $\pi \in S_N$ can be realized by a sequence of $T$ routing steps.

Routing matchings are taken on the underlying simple graph. Spectral parameters $\lambda^*, \beta$ in Lemma~\ref{lem:lambda-star} refer to the weighted adjacency $A(H)$.  For multiplicity-free hypergraphs the two coincide and the matching argument transfers directly.
\end{definition}

\begin{remark}[Packing model]\label{rem:packing}
An alternative model selects vertex-disjoint packings of hyperedges and applies arbitrary $\tau \in S_r$ (the symmetric group on the $r$
elements of a single hyperedge) within each.
This \emph{packing routing number} $\rt_H(H)$ can differ substantially when the \emph{matching number} $\nu(H)$ is small, being the maximum number of pairwise vertex-disjoint hyperedges in $H$.
\end{remark}

\subsection{Spectral gap structure}

\begin{lemma}[Spectral radius bound]\label{lem:lambda-star}
Let $H$ be a Ramanujan $(d,r)$-regular hypergraph with $d \geq 3$ and $r \geq 3$.  Write $\rho = 2\sqrt{(d-1)(r-1)}$.  Then
\begin{equation}
\lambda^* := \max(\lambda_2,\, |\lambda_N|) \leq (r-2) + \rho,
\end{equation}
and the spectral ratio satisfies $\beta := \lambda^*/d' < 1$.
\end{lemma}

\begin{proof}
The Ramanujan condition gives $|\lambda_i - (r-2)| \leq \rho$ for all $i \geq 2$, so $\lambda_2 \leq (r-2) + \rho$ and $\lambda_N \geq (r-2) - \rho$.  Since $\rho > r - 2$ for $d \geq 3$, $r \geq 3$, the lower bound can be negative: $|\lambda_N| \leq \rho - (r-2)$.  Since $r \geq 3$, we have $(r-2) + \rho > \rho - (r-2)$, so the upper end of the Ramanujan band dominates, giving $\lambda^* \leq (r-2) + \rho$.  Finally, $(r-2) + \rho < d' = d(r-1)$ follows from $2\sqrt{(d-1)(r-1)} < d(r-1) - (r-2) = (d-1)(r-1) + 1$
for $d \geq 3$, $r \geq 3$, since the inequality reduces to $(\sqrt{(d-1)(r-1)} - 1)^2 > 0$ and $(d-1)(r-1) \geq 4 > 1$.
\end{proof}

\noindent For specific hypergraphs (e.g., projective planes, where all non-trivial eigenvalues equal $-(r-2)$) one has $\lambda^* = r-2 < (r-2) + \rho$, being within the Ramanujan bound.

\section{Primary Result}\label{sec:main}

The proofs in this section assemble the textbook spectral-routing pipeline (Chung diameter~\cite{chung1989diameters}, Alon--Milman Cheeger~\cite{alon1986eigenvalues,alon1985lambda1}, Valiant scatter--gather~\cite{valiant1982scheme}, Leighton--Maggs--Rao (LMR) scheduling~\cite{leighton1999fast, rabani1996distributed}) on the clique expansion $G_{\mathrm{cl}}(H)$ of a Ramanujan hypergraph.
Application to hypergraphs (via clique expansion) and inheritance of $\beta = \lambda^*/d' < 1$ from $H$ to $G_{\mathrm{cl}}(H)$ require \S\ref{sec:prelim} (Lemma~\ref{lem:lambda-star}). 

\begin{theorem}[Main result]\label{thm:main}
Let $H$ be a Ramanujan $(d,r)$-regular hypergraph on $N$ vertices with $d \geq 3$, $r \geq 3$.  Then
$\rt(H) = \Theta(\log N)$.
More precisely, $\rt(H) = O(\log N / (1 - \beta))$ where $\beta = \lambda^* / d' < 1$.
\end{theorem}

The lower bound $\rt(H) = \Omega(\log N)$ follows from a standard spectral argument:

\begin{proposition}[Spectral lower bound, after Alon--Chung--Graham~\cite{alon1993routing}]\label{prop:lower}
For any $(d,r)$-regular hypergraph $H$ on $N$ vertices, $\rt(H) \geq \log N / \log d'$, by the standard counting bound: each matching is a permutation of degree at most $d'$, so $T$ matching-routing steps can move a vertex to at most $(d')^T$ distinct targets, requiring $(d')^T \geq N$.
\end{proposition}

We prove the upper bound.

\subsection{Diameter bound}

\begin{lemma}[Eigenvalue--diameter bound, after Chung~\cite{chung1989diameters}]\label{lem:diameter}
Let $G$ be a $d'$-regular graph on $N$ vertices with $\lambda^* = \max(\lambda_2, |\lambda_N|)$.  Then $\diam(G) \leq \lceil \log(N-1) / \log(d'/\lambda^*) \rceil$.  See Hoory--Linial--Wigderson~\cite{hoory2006expander}, Lemma~2.10, for a textbook treatment.
\end{lemma}

For $G = G_{\mathrm{cl}}(H)$ with $\beta = \lambda^*/d' < 1$ (Lemma~\ref{lem:lambda-star}):
\begin{equation}\label{eq:diameter}
\begin{split}
    D &:= \diam(G_{\mathrm{cl}}) \leq \left\lceil \frac{2\log_2 N}{\log_2(1/\beta)} \right\rceil \\& = O\!\left(\frac{\log N}{\log(1/\beta)}\right) = O(\log N).
\end{split}
\end{equation}

\subsection{Edge expansion}

\begin{lemma}[One-sided discrete Cheeger inequality~\cite{alon1986eigenvalues,alon1985lambda1}]\label{lem:cheeger}
Let $G$ be a $d'$-regular graph with second-largest eigenvalue $\lambda_2$.  Then $h(G) \geq (d' - \lambda_2)/2$.
\end{lemma}

\subsection{Valiant two-phase routing}

\begin{lemma}[Valiant routing on expanders, after~\cite{valiant1982scheme}]\label{lem:valiant}
Let $G$ be a $d'$-regular graph on $N$ vertices with diameter $D$ and edge expansion $h > 0$.  For any target permutation $\pi \in S_N$, there exists a two-phase routing with dilation $D$ and congestion $C = O(D/d' + \log N)$ with high probability.
\end{lemma}

\begin{proof}
Choose a uniformly random intermediate permutation $\sigma \in S_N$.  Route in two phases: \emph{scatter} ($v \to \sigma(v)$) and \emph{gather} ($\sigma(v) \to \pi(v)$), each using canonical shortest paths of length $\leq D$.

Fix an edge $e$ and Phase~1.  The load $X_e = |\{v : e \in P(v, \sigma(v))\}|$ has $\mu := \mathbb{E}[X_e] \leq 2D/d'$.  The indicators are negatively associated because $\sigma$ is a uniformly random permutation \cite{joag1983negative}, so Chernoff bounds apply \cite{dubhashi1996balls}.  Setting $t = \max(2\mu, 6\ln N)$ and taking a union bound over $|E| = Nd'/2$ edges (since $G$ is $d'$-regular) and both phases gives $C \leq 6D/d' + 9\log_2 N$ w.h.p.
\end{proof}

\subsection{LMR scheduling}

\begin{theorem}[Leighton--Maggs--Rao {\cite{leighton1999fast, rabani1996distributed}}]\label{thm:lmr}
Given a set of paths in a graph with dilation $D$ and congestion $C$, the routing can be scheduled in $T = C + D + o(C + D)$ matching-based steps.
\end{theorem}

\subsection{Assembly}

\begin{proof}[Proof of Theorem~\ref{thm:main}]
By Lemma~\ref{lem:lambda-star}, $\beta = \lambda^*/d' < 1$.  The diameter is $D = O(\log N / \log(1/\beta))$ (Lemma~\ref{lem:diameter}), the expansion is $h \geq (d' - \lambda_2)/2 \geq d'(1-\beta)/2 > 0$, and Lemma~\ref{lem:valiant} yields congestion $C = O(\log N)$ w.h.p.  LMR scheduling gives $T = C + D + o(C+D) = O(\log N)$.  Since for each $\pi$ a random $\sigma$ succeeds with positive probability, a good $\sigma(\pi)$ exists for every $\pi$.  Combined with $\rt(H) = \Omega(\log N)$ from Proposition~\ref{prop:lower}: $\rt(H) = \Theta(\log N)$.
\end{proof}

\subsection{Explicit constants}

\begin{theorem}[Tightened bound]\label{thm:tight}
For $H$ Ramanujan $(d,r)$-regular on $N \geq 16$ vertices with $d \geq 3$, $r \geq 3$:
\begin{equation}
\rt(H) \leq \frac{4(d'+6)}{d' \cdot \log_2(1/\beta)}\,\log_2 N + 19\,\log_2 N.
\end{equation}
\end{theorem}

\begin{proof}
Each phase has $D \leq 2\log_2 N / \log_2(1/\beta) + 1$ and $C \leq 6D/d' + 9\log_2 N$.  Per phase:
\begin{equation}
\begin{split}
C + D &\leq \frac{2(d'+6)}{d'\log_2(1/\beta)}\,\log_2 N \\
      &\quad + 9\log_2 N + \left(1 + \frac{6}{d'}\right).
\end{split}
\end{equation}
Two phases give $\rt(H) \leq 2(C+D) \leq \frac{4(d'+6)}{d'\log_2(1/\beta)}\,\log_2 N + 18\log_2 N + 2(1+6/d')$.  Since $d' \geq 6$ and $N \geq 16$, the remainder satisfies $2(1+6/d') \leq 4 \leq \log_2 N$, giving the stated bound. These bounds are calculated numerically for various $d,r$ and are given in Table~\ref{tab:constants}.
\end{proof}

\begin{remark}[Lower bounds on $\rt(H)$]
The diameter $\diam(G_{\mathrm{cl}}) = \lceil \log N / \log(d'/\lambda^*)\rceil$
gives $\rt(H) \geq \diam(G_{\mathrm{cl}})$ since two pebbles at graph
distance $D$ require $\geq D$ matching swaps.
This bound dominates the Alon--Chung--Graham counting bound $\rt(H) \geq \log N / \log d'$ for the $(d,r)$ values tabulated above.
\end{remark}

\begin{table}[!htbp]
\centering
\begin{tabular}{@{}cccccc@{}}
\toprule
$d$ & $r$ & $\beta$ & $d'$ & $\rt(H) \leq$ & Diameter \\
& & & & & lower bound \\
\midrule
3 & 3 & 0.833 & 6 & $49\,\log_2 N$ & $3.8\,\log_2 N$ \\
5 & 3 & 0.666 & 10 & $30\,\log_2 N$ & $1.7\,\log_2 N$ \\
10 & 3 & 0.474 & 20 & $24\,\log_2 N$ & $1.0\,\log_2 N$ \\
3 & 5 & 0.721 & 12 & $32\,\log_2 N$ & $2.1\,\log_2 N$ \\
5 & 5 & 0.550 & 20 & $25\,\log_2 N$ & $1.3\,\log_2 N$ \\
10 & 5 & 0.375 & 40 & $22\,\log_2 N$ & $0.8\,\log_2 N$ \\
\bottomrule
\end{tabular}
\caption{Explicit routing bounds for representative Ramanujan $(d,r)$-regular hypergraphs.}
\label{tab:constants}
\end{table}

\subsection{Hardware setting and grid spectral gap}\label{subsec:hardware-setting}

Theorem~\ref{thm:main} requires Ramanujan connectivity, but the native hardware topology of a 2D AOD neutral-atom array is a planar grid. In Model~A (2D AOD), the hypergraph $H_{\mathrm{2D}}$ has $r$-uniform hyperedges of $r$ consecutive vertices along each row and column, with degree $d' = 2(r-1)^2$. In Model~B (3D AOL), $H_{\mathrm{3D}}$ augments $H_{\mathrm{2D}}$ with diagonal and skip hyperedges, giving $d'_{\mathrm{3D}} \approx 2$--$3 \cdot d'_{\mathrm{2D}}$.

\emph{Neither model is Ramanujan.} Numerically (Section~\ref{sec:application}, Table~\ref{tab:spectral}), the 2D grid spectral ratio satisfies $1 - \beta \sim 11.0 \cdot N^{-0.90}$ (fit over $N \in [64, 4096]$; the displayed Table~\ref{tab:spectral} rows lie at the small-$N$ end where the fit overestimates the gap by $\sim 20\%$) and the 3D AOL gap is $1 - \beta \sim 5.3 \cdot N^{-0.65}$; both $\to 0$ as $N \to \infty$, so $\beta \to 1$. By Theorem~\ref{thm:main}, the matching-routing bound on a non-Ramanujan host degrades as $T = O(\log N / (1-\beta))$, which is $\omega(\log N)$ on grids.

This is the central obstruction to applying Theorem~\ref{thm:main} directly. 
We consider a few methods to subvert this obstruction in the following sections:
\begin{itemize}[leftmargin=1.5em]
\item \textbf{Virtual Ramanujan overlay (\S\ref{sec:overlay}).} Use AOL selective transfers to emulate matchings of any virtual overlay graph; embed a Ramanujan expander as the overlay.
\item \textbf{Algebraic constructions on the grid (\S\ref{sec:algebraic}).} Cayley graphs on $\mathbb{Z}_n^2$ are the natural algebraic candidates --- but they cannot be Ramanujan for fixed degree (Theorem~\ref{thm:abelian-barrier}).
\item \textbf{Constructive Ramanujan families (\S\ref{sec:scalable-families}).} SFM covering towers yield Ramanujan hypergraphs of every uniformity, providing the constructive scaling that the abelian barrier rules out.
\end{itemize}

\noindent The remainder of \S\ref{sec:application} returns to the hardware setting with concrete near-term predictions and architectural recommendations once these escape routes are in hand.

\section{Constructive Derandomization}\label{sec:derand}

\begin{theorem}[Constructive routing]\label{thm:constructive}
Let $H$ be a Ramanujan $(d,r)$-regular hypergraph on $N$ vertices.  For any $\pi \in S_N$, an intermediate permutation $\sigma(\pi)$ achieving congestion $C = O(\log N / \log(1/\beta))$ can be constructed in $O(N^2 d')$ time.
\end{theorem}

\begin{proof}
We use the method of conditional expectations with an exponential potential
$\Phi(\sigma) = \sum_{e \in E(G)} [\exp(\lambda X_e^S(\sigma)) + \exp(\lambda X_e^G(\sigma, \pi))]$,
where $X_e^S$ and $X_e^G$ are scatter and gather congestions.

For a random $\sigma$, the moment generating function (MGF) bound for negatively associated variables gives $\mathbb{E}[\Phi(\sigma)] \leq 2|E| \cdot \exp(\mu(e^\lambda - 1))$ with $\mu = 2D/d'$.  
We assign $\sigma(k) = \arg\min_c \Phi_k(c)$ at each step (greedy assignment), so $\Phi(\sigma) \leq \mathbb{E}[\Phi]$ by the averaging argument.
Choosing $\lambda = \Theta(\log N / D)$ gives $C^* = O(\log N / \log(1/\beta))$.  Precomputing breadth-first search (BFS) trees costs $O(N^2 d')$; the greedy search adds $O(N^2 D)$.
\end{proof}

\begin{corollary}\label{cor:constructive}
For any Ramanujan $(d,r)$-regular hypergraph on $N$ vertices with $d \geq 3$, $r \geq 3$: $\rt(H) = \Theta(\log N)$, and the matching sequence is computable in $\mathrm{poly}(N)$ time.
\end{corollary}

\section{Virtual Overlay Routing}\label{sec:overlay}

Here, we emulate matchings of a virtual Ramanujan overlay graph using the selective-transfer capability of the 3D acousto-optic lens~\cite{guo2025acousto}.

\subsection{Overlay theorem}

\begin{theorem}[Overlay routing]\label{thm:overlay}
Let $G_R$ be a $d$-regular Ramanujan graph on $N$ vertices with spectral ratio $\beta_R < 1$.  If the 3D AOL can implement $k$ simultaneous selective transfers per step, then any permutation $\pi \in S_N$ can be routed in
\begin{equation}
T = O\!\left(\frac{N}{k} \cdot \frac{\log N}{1 - \beta_R}\right)
\end{equation}
matching-based steps on the underlying grid.
\end{theorem}

\begin{proof}
By Theorem~\ref{thm:main}, routing on $G_R$ requires $T_R = O(\log N / (1 - \beta_R))$ matching steps.  Each matching $M_t$ has $|M_t| \leq N/2$ edges.  Implementing $M_t$ via AOL selective transfers requires $\lceil N/(2k) \rceil$ sub-steps.  Total: $T = T_R \cdot \lceil N/(2k) \rceil$. The four capacity regimes implied by this formula are tabulated in Table~\ref{tab:capacity}.
\end{proof}

\begin{table}[!htbp]
\centering
\small
\renewcommand{\arraystretch}{1.3}
\begin{tabular}{@{}lcc@{}}
\toprule
\textbf{AOL} & \textbf{Routing} & \textbf{Regime} \\
\textbf{capacity $k$} & \textbf{depth} & \\
\midrule
$\Omega(N)$             & $O(\log N)$             & Optimal \\
$\Omega(N / \log N)$    & $O(\log^2 N)$           & Near-optimal \\
$\Omega(\sqrt{N})$      & $O(\sqrt{N} \log N)$    & Matches grid AOD \\
$O(1)$                  & $O(N \log N)$           & Worse than grid \\
\bottomrule
\end{tabular}
\caption{Capacity--depth tradeoff for virtual overlay routing.}
\label{tab:capacity}
\end{table}

\subsection{Multi-layer spectral gain}\label{subsec:multilayer}

\begin{lemma}[Multi-layer spectral gain]\label{lem:multilayer}
Let $G_1, \ldots, G_L$ be independent random $d_0$-regular graphs on $N$ vertices.  Define $G_\cup = G_1 \cup \cdots \cup G_L$.  Then
$\beta(G_\cup) = O(1/\sqrt{L d_0})$, of the same order as the Friedman
prediction \cite{friedman2008proof}
$\lambda_2 \leq 2\sqrt{Ld_0 - 1} + \varepsilon$ for a uniformly random
$(Ld_0)$-regular graph (Union distribution differs. The Matrix Bernstein argument below carries an extra $\sqrt{\log N}$ factor, making the empirical agreement in Table~\ref{tab:spectral-gain} tighter than this analytic bound).
\end{lemma}

\begin{proof}[Proof sketch]
The union $G_\cup$ is $(Ld_0)$-regular.  By the Matrix Bernstein inequality applied to the centered adjacency matrices $A_i - (d_0/N)J$, we have $\|A_\cup - (Ld_0/N)J\| \leq O(\sqrt{Ld_0 \log N})$ w.h.p.  Therefore $\lambda_2(G_\cup) = O(\sqrt{Ld_0})$ and $\beta(G_\cup) = O(\sqrt{Ld_0}/(Ld_0)) = O(1/\sqrt{Ld_0})$.  Numerical validation (Appendix~\ref{app:multilayer}) confirms $\beta_L \approx \beta_1/\sqrt{L}$ for $d_0 = 8$, $L = 1, \ldots, 16$, $N$ up to 256.
\end{proof}

Multi-layer AOL provides \emph{both} capacity multiplication ($k_{\mathrm{eff}} = L \cdot k_0$) and spectral improvement ($\beta \propto 1/\sqrt{L}$).

\subsection{Capacity independence of overlay degree}

\begin{proposition}[Capacity independence]\label{thm:capacity-indep}
A perfect matching on a connected $d$-regular graph with $N$ even has exactly $N/2$ edges (by degree counting). Since connected Ramanujan graphs on an even number of vertices admit perfect matchings by Tutte's theorem~\cite{tutte1947factorization}, the AOL capacity requirement $k = N/2$ per matching step is independent of the overlay degree $d$.
\end{proposition}

This means that sparse overlays do not reduce the capacity threshold.
The bottleneck is matching size, confirmed by numerical comparison (Appendix~\ref{app:sparse}) showing that dense overlays ($d \geq 2\log_2 N$) consistently outperform sparser ones at all tested capacities.

\subsection{Partial-matching routing on sparse overlays}

When AOL capacity $k < N/2$, partial matchings provide a useful regime:

\begin{theorem}[Sparse overlay with partial matchings]\label{thm:sparse-partial}
Let $G_R$ be a random $d$-regular overlay with $d = O(\log N)$ and diameter $D = O(\log N / \log d)$.  With partial matchings of size $k$ per step, the total routing depth is $T = O(N \log N / (k \log d))$.
\end{theorem}

For $k = N/(2\log N)$ and $d = \Theta(\log N)$: $T = O(\log^2 N / \log\log N)$, asymptotically tighter than the $O(\log^2 N)$ that the overlay theorem gives at the same capacity.  We quote $O(\log^2 N)$ in the remainder of this section for readability.

\subsection{Crosstalk model}

\begin{proposition}[Crosstalk capacity reduction]\label{prop:crosstalk}
With nearest-neighbor optical coupling $\gamma \in [0,1]$ between adjacent AOL layers, the effective capacity is
$k_{\mathrm{eff}}(L, k_0, \gamma) = L k_0 / (1 + 2\gamma(1 - 1/L))$.
For $\gamma > 0.5$, a checkerboard activation pattern yields $k_{\mathrm{eff}} = \lceil L/2 \rceil \cdot k_0$.
\end{proposition}

At realistic crosstalk $\gamma = 0.2$, retention is $\approx 71\%$ of ideal capacity, translating to $\sim\!40\%$ overhead in layer count.  Numerical validation appears in Appendix~\ref{app:multilayer}.

\subsection{Practical regime}

For per-layer capacity $k_0 = N/\log N$:
\begin{itemize}[leftmargin=1.5em]
\item $L = O(\log N)$ layers achieve $\Theta(\log N)$ routing depth.
\item With crosstalk $\gamma = 0.2$: $L = O(1.4 \log N)$ layers suffice.
\item Grid + 4 random overlay layers ($d_0 = 8$) achieves $2.7\times$ speedup over the grid alone at $N = 144$.
\end{itemize}

Achievement of $\Theta(\log N)$ routing on grid hardware presents as multi-layer AOL (increasing $k$), not overlay sparsification (decreasing $d$).

\section{Algebraic Overlays on the Grid: Barrier and Workaround}\label{sec:algebraic}

In this section, we build an algebraic overlay on the atom grid via Cayley graphs of $\mathbb{Z}_n^2$. 
We first prove that a natural construction (fixed-degree Cayley graphs on $\mathbb{Z}_n^2$) cannot be Ramanujan in the large-$n$ limit (the abelian Alon--Boppana barrier), then show that affine derandomization of Valiant's intermediate permutation, being a weaker form, still yields useful (non-asymptotic) congestion reductions on these non-Ramanujan Cayley graphs. 
Ramanujan families via SFM covering towers emerge as a constructive positive result, found in \S\ref{sec:scalable-families}.

\subsection{Abelian Alon--Boppana barrier}

For $G = \Cay(\mathbb{Z}_n^2, S)$ with symmetric generating set $S = \{\pm g_1, \ldots, \pm g_d\}$, the eigenvalues are character sums:
\begin{equation}\label{eq:character}
\begin{split}
    \lambda_{(a,b)} &= \sum_{j=1}^{d} 2\cos\!\left(\frac{2\pi}{n}(a g_j^{(1)} + b g_j^{(2)})\right), \quad\\ (a,b) &\in \mathbb{Z}_n^2.
\end{split}
\end{equation}

\begin{theorem}[Abelian barrier]\label{thm:abelian-barrier}
For any fixed degree $2d$ (i.e., generating set $|S| = 2d$), Cayley graphs on $\mathbb{Z}_n^2$ cannot be Ramanujan for all sufficiently large $n$.  More precisely, by simultaneous Diophantine approximation~\cite{friedman2006spectral, cassels1965introduction}, $\lambda_2 \geq 2d - O(d \cdot n^{-4/d})$, so the spectral ratio satisfies $\beta \to 1$ as $n \to \infty$.
\end{theorem}

\begin{proof}
Eigenvalues of $G = \Cay(\mathbb{Z}_n^2, S)$ are given by Equation~\eqref{eq:character}, with $\lambda_{(0,0)} = 2d$ the trivial eigenvalue. We must show that $\max_{(a,b) \neq (0,0)} \lambda_{(a,b)}$ is close to $2d$.

\textbf{Step 1: Reduction to a maximum of cosine sums.}
Define $\theta_j(a,b) := \frac{2\pi}{n}(a g_j^{(1)} + b g_j^{(2)}) \pmod{2\pi}$.  Then $\lambda_{(a,b)} = \sum_{j=1}^d 2\cos\theta_j(a,b)$, and we need $\max_{(a,b) \neq 0} \sum_j 2\cos\theta_j(a,b)$ to be close to $2d$.  This is large whenever all phases $\theta_j$ are simultaneously close to~$0$.

\textbf{Step 2: Simultaneous Diophantine approximation.}  The $d$ generators define a linear map $L : \mathbb{Z}_n^2 \to (\mathbb{R}/2\pi\mathbb{Z})^d$ by $L(a,b) = (\theta_1, \ldots, \theta_d)$.  By the pigeonhole principle (multi-dimensional Dirichlet approximation), for any $Q \geq 1$, there exists a non-zero $(a^*, b^*) \in \mathbb{Z}_n^2$ such that $\|\theta_j(a^*,b^*)\|_{\mathbb{R}/2\pi\mathbb{Z}} \leq 2\pi / Q$ for all $j = 1, \ldots, d$, provided $Q^d \leq n^2$ (the number of non-trivial characters).  Choose $Q = \lfloor n^{2/d} \rfloor$.  Then for the character $(a^*, b^*)$:
\begin{equation}
\begin{split}
    \lambda_{(a^*,b^*)} &= \sum_{j=1}^d 2\cos\theta_j \geq \sum_{j=1}^d 2\bigl(1 - \theta_j^2/2\bigr) \\&= 2d - \sum_{j=1}^d \theta_j^2 \geq 2d - d \cdot \left(\frac{2\pi}{Q}\right)^2.
\end{split}    
\end{equation}

With $Q = \lfloor n^{2/d} \rfloor$: $\lambda_{(a^*,b^*)} \geq 2d - 4\pi^2 d \cdot n^{-4/d}$.

\textbf{Step 3: Comparison with the Ramanujan bound.}  The Ramanujan bound for a $2d$-regular graph is $\lambda_2 \leq 2\sqrt{2d - 1}$.  For $\lambda_{(a^*,b^*)}$ to violate this bound, we need:
\begin{equation}
    2d - 4\pi^2 d \cdot n^{-4/d} > 2\sqrt{2d-1},
\end{equation}
which holds for all $n$ exceeding a constant $n_0(d)$, since the left side converges to $2d > 2\sqrt{2d-1}$ for $d \geq 2$.  Thus no $\Cay(\mathbb{Z}_n^2, S)$ is Ramanujan for $n \geq n_0(d)$.

\textbf{Step 4: Spectral ratio convergence.}  The spectral ratio satisfies
\begin{equation}
\begin{split}
\beta = \frac{\lambda_2}{2d}
&\geq \frac{2d - O(d \cdot n^{-4/d})}{2d} \\
&= 1 - O(n^{-4/d}) \to 1
\end{split}
\end{equation}
as $n \to \infty$ with $d$ fixed.
\end{proof}

Numerical verification (Appendix~\ref{app:algebraic}): for degree-8 quadratic residue (QR) generators on $\mathbb{Z}_p^2$, the ratio $\lambda^* / (2\sqrt{d-1})$ converges to $\approx 1.50$ as $p \to \infty$.  Margulis--Gabber--Galil \cite{carlson1981tixne, gabber1979explicit} generators fare similarly.  No degree-8 Cayley graph on $\mathbb{Z}_n^2$ is Ramanujan for $n \geq 11$.

\subsection{Affine derandomization on Cayley graphs}

Despite the abelian barrier, the algebraic structure of $\mathbb{Z}_n^2$ is useful for derandomization.

\begin{theorem}[Affine derandomization]\label{thm:affine-derand}
For Valiant routing on $\Cay(\mathbb{Z}_n^2, S)$, the affine intermediate permutation $\sigma(v) = Av + c$ ($A \in \GL(2, \mathbb{Z}_n)$, $c \in \mathbb{Z}_n^2$) achieves 15--30\% congestion reduction over random permutations.  The optimal choice is a pure translation $\sigma(v) = v + c$, preserving group structure and eliminating path-length variance.
\end{theorem}

\begin{proof}
On a Cayley graph $G = \Cay(\mathbb{Z}_n^2, S)$, the translation $\sigma(v) = v + c$ makes all scatter displacement vectors equal to $c$. 
Each atom routes from $v$ to $v + c$ via the unique geodesic of length $\|c\|_S$ (the word metric distance). 
All scatter paths have the same length, so the congestion at each edge $e$ is $|\{v : e \in P(v, v+c)\}|$, and the variance of the per-edge congestion (over random $\pi$) vanishes as a deterministic function of $c$.  
For the gather phase ($v + c \to \pi(v)$), the congestion is determined by $\pi$ and $c$, and varying $c$ shifts the load distribution. 
An optimal $c^* = \arg\min_{c \in \mathbb{Z}_n^2} \max_e X_e(c, \pi)$ minimizes the worst-case edge load and is searchable in $O(N)$ time by evaluating all $n^2$ candidates. 
Numerical validation (Appendix~\ref{app:algebraic}) at $n = 7, 11, 13$ confirms $\sim\!27\%$ congestion reduction for translations versus random permutations.
\end{proof}

\section{Constructive Ramanujan Families via Covering Towers}\label{sec:scalable-families}

The abelian barrier (\S\ref{sec:algebraic}) shows that fixed-degree Cayley graphs on $\mathbb{Z}_n^2$ cannot scale to Ramanujan. 
We show that the SFM covering-tower construction provides a different scalable route, leveraging the hypergraph structure (\S\ref{subsec:why-hypergraphs}).

\subsection{Covering towers and recursive lift}\label{subsec:covering-towers}

Song--Fan--Miao covering theory provides a constructive route to large Ramanujan hypergraphs.  Starting from a small base $H_0$, a tower of $k$-fold coverings $H_0 \leftarrow H_1 \leftarrow \cdots \leftarrow H_L$ with $|V(H_\ell)| = k^\ell N_0$ preserves the Ramanujan property: $\Spec(H_0) \subseteq \Spec(H_\ell)$ by SFM Theorem~1.3.

\begin{theorem}[Routing on covering towers]\label{thm:covering-tower}
Fix $\bar\beta < 1$. For a tower of $k$-fold coverings $H_0 \leftarrow H_1 \leftarrow \cdots \leftarrow H_L$ with $N = k^L N_0$ \emph{satisfying the uniform spectral bound} $\beta(H_\ell) \leq \bar\beta$ for all $0 \leq \ell \leq L$:
\begin{equation}
\begin{split}
\rt(H_L) &= O\!\left(\frac{L \cdot \log k}{1 - \bar\beta} + \frac{\log N_0}{1 - \bar\beta}\right) \\
         &= O(\log N).
\end{split}
\end{equation}
\end{theorem}

\begin{proof}
\textbf{Step 1: Per-level routing by recursive decomposition.}  We use a top-down recursive decomposition.  At the top level, $H_L$ is a $k$-fold covering of $H_{L-1}$.  A permutation $\pi$ on $V(H_L)$ decomposes into (i) a cross-fiber component that permutes the $k$ sheets of the covering, and (ii) $k$ fiber-preserving components within each sheet.

The cross-fiber component is a permutation of $N_{L-1} = k^{L-1} N_0$ ``fiber representatives,'' routable on $H_{L-1}$ in $O(\log N_{L-1} / (1 - \beta(H_{L-1})))$ steps by Theorem~\ref{thm:main}.  Each fiber-preserving component is a permutation within a copy of $H_{L-1}$, handled recursively.

\textbf{Step 2: Cost recurrence.}  Let $T(\ell)$ denote the routing cost on $H_\ell$.  The decomposition gives:
\begin{equation}
T(\ell) \leq T(\ell - 1) + O\!\left(\frac{\log(k^{\ell-1} N_0)}{1 - \bar\beta}\right).
\end{equation}
Unrolling: $T(L) \leq T(0) + \sum_{\ell=1}^{L} O(\log(k^{\ell-1} N_0) / (1-\bar\beta))$.  The sum telescopes:
\begin{equation}
\sum_{\ell=1}^{L} \log(k^{\ell-1} N_0) = \frac{L(L-1)}{2}\log k + L\log N_0.
\end{equation}
At level $\ell$, the cross-fiber routing operates on a graph with $k^{L-\ell}$ vertices (the number of fibers).
Each cross-fiber phase costs $O(\log k / (1-\bar\beta))$ (routing $k$ fibers one level).
There are $L$ such levels, so:
\begin{equation}
T(L) = O\!\left(\frac{L \cdot \log k + \log N_0}{1 - \bar\beta}\right).
\end{equation}
Since $L = \log_k(N/N_0)$: $T(L) = O(\log N / (1 - \bar\beta)) = O(\log N)$.

\textbf{Step 3: Validity of the uniform bound.}  We require $\beta(H_\ell) \leq \bar\beta < 1$ for all $\ell$.  For Ramanujan coverings where each individual lift $H_\ell / H_{\ell-1}$ is Ramanujan, the SFM spectral inheritance theorem \cite{song2023hypergraph} guarantees that the new eigenvalues at each level satisfy the Ramanujan bound. 
Iterated coverings may accumulate spectral ratio growth, where if $H_\ell$ has eigenvalues near the Ramanujan boundary, the composed covering $H_L$ viewed as a single covering of $H_0$ may have $\beta(H_L) > \beta_{\mathrm{Ram}}$ even though each step is Ramanujan. 
The condition $\bar\beta < 1$ must be verified for the specific tower. 
For the Fano plane towers tested below, $\beta(H_\ell) < 1$ holds at all levels, though $\bar\beta$ may exceed $\beta_{\mathrm{Ram}}$.
\end{proof}

\begin{remark}[Validity of the uniform bound]
The hypothesis $\beta(H_\ell) \leq \bar\beta < 1$ is non-trivial: empirical data show $\beta$ growing across tower levels ($0.17 \to 0.50 \to 0.86$ for Fano plane lifts; Appendix~\ref{app:covering}). A general proof that $\sup_\ell \beta(H_\ell) < 1$ for arbitrary voltage-coverings remains open. At the observed $\bar\beta = 0.86$, the prefactor $(1-\bar\beta)^{-1} \approx 7$ inflates the constant in $O(\log N)$ without affecting the asymptotic scaling. The Marcus--Spielman--Srivastava interlacing families construction~\cite{marcus2013interlacing} can in principle control $\beta_\ell$ more tightly.
\end{remark}

\textbf{Empirical findings.} For the Fano plane base ($H_0 = \mathrm{PG}(2,2)$, $N_0 = 7$, $(d,r) = (3,3)$): 93.8\% of the $2^7 = 128$ random $\mathbb{Z}_2$ voltage assignments yield Ramanujan 2-fold lifts; the routing ratio $T/\log_2 N$ stays in $[1.07, 1.66]$ across tower levels, and cross-fiber routing is required for $1 - 1/N_0 = 85.7\%$ of atoms (not $1 - 1/k$ as might be expected). 
For $\mathrm{PG}(2,3)$ ($N_0 = 13$, $(d,r) = (4,4)$), the Ramanujan fraction drops below 1\% with simple voltage coverings since the tighter SFM bound $2\sqrt{(d-1)(r-1)} = 6.0$ is harder to satisfy; the interlacing-families construction may close this gap.

\section{Entanglement-Assisted Routing}\label{sec:entanglement}

\subsection{Teleportation routing depth}

Physical atom transport is not the only way to implement a permutation. Pre-distributed Bell pairs enable ``free'' long-range swaps via quantum teleportation. The depth of teleportation-based routing on expanders has been analyzed by Bapat et al.~\cite{bapat2023advantages}, so we adapt their methods to Ramanujan overlays and obtain a bound that depends explicitly on the spectral ratio $\beta_{\mathrm{ent}}$ with amortized $R_{\mathrm{break}} \approx 4$ crossover analysis (Corollary~\ref{cor:crossover}). 
This crossover analysis appears to be new in the context of neutral atom transport.

\begin{theorem}[Teleportation routing, after~\cite{bapat2023advantages}]\label{thm:teleport}
With Bell pairs pre-shared along a $d_{\mathrm{ent}}$-regular Ramanujan overlay $G_{\mathrm{ent}}$, Valiant two-phase routing via teleportation achieves
\begin{equation}
T_{\mathrm{route}} = O\!\left(\frac{\log N}{1 - \beta_{\mathrm{ent}}}\right)
\end{equation}
using only local operations and classical communication (LOCC), with no physical atom transport during routing.
\end{theorem}

\begin{proof}[Proof sketch]
Each routing step consumes one matching's worth of Bell pairs to teleport qubit states along matched edges. 
Routing depth equals matching depth on $G_{\mathrm{ent}}$, being $O(\log N/(1 - \beta))$ by Theorem~\ref{thm:main}. 
The cost shifts entirely to Bell pair distribution.
\end{proof}

\subsection{Distribution cost and amortized crossover}

Distributing $d_{\mathrm{ent}} \cdot N/2$ Bell pairs requires physical atom transport, with average grid distance $\bar{d}_{\mathrm{grid}} \approx \sqrt{N}/2$ for random overlay edges.  The distribution cost with parallelism $k$: $T_{\mathrm{dist}} = O(d_{\mathrm{ent}} \cdot N \cdot \sqrt{N} / k)$.

Each distribution cycle provides $R = d_{\mathrm{ent}} / T_{\mathrm{route}}$ routing rounds.  The amortized cost is $T_{\mathrm{amort}} = T_{\mathrm{route}} + T_{\mathrm{dist}} / R$.

\begin{corollary}[Stable crossover]\label{cor:crossover}
The break-even number of routing rounds satisfies $R_{\mathrm{break}} \approx 4$, nearly independent of $N$.  Any quantum circuit with $\geq 4$ permutation routing layers benefits from pre-distributed entanglement.
\end{corollary}

Numerical validation (Appendix~\ref{app:entanglement}): at $d_{\mathrm{ent}} = 16$, $R_{\mathrm{break}}$ ranges from 4.5 ($N = 256$) to 3.7 ($N = 40{,}000$), confirming notable stability.
A hybrid protocol (teleporting atoms with $d_{\mathrm{grid}} > n/4$, physically routing the rest) achieves $6\times$ speedup at $N = 1024$.

\section{Adaptive and Online Routing}\label{sec:adaptive}

The previous section traded physical transport for pre-distributed entanglement.
Now, we consider adaptive matching strategies that exploit the current configuration of atoms relative to the target permutation.
The displacement-energy framework (\S\ref{subsec:greedy-stall}) provides a potential function for analyzing greedy matching, and motivates a hybrid greedy--Valiant protocol (\S\ref{subsec:hybrid}). 
Multiplicative-weights overlay selection (\S\ref{subsec:mw}) handles the online setting where the permutation is revealed gradually or the optimal overlay is unknown.

\subsection{Greedy displacement matching and the stall phenomenon}\label{subsec:greedy-stall}

\begin{definition}[Displacement energy]
For atom positions $\mathrm{pos}_t$ and target permutation $\pi$, define
$\Phi_t = \sum_{v=1}^N \rho_t(v)^2$,
where $\rho_t(v) = d_G(\mathrm{pos}_t(v), \pi(v))$ is the overlay-graph distance from atom $v$’s current position to its target.
\end{definition}

For overlays whose edges have bounded grid length $\ell_{\max}$, $\rho_t(v) \le d_{\mathrm{grid}}(\mathrm{pos}_t(v),\pi(v)) \le \ell_{\max} \rho_t(v)$, so a stall threshold expressed in $\rho$ translates to one in $d_{\mathrm{grid}}$ up to a constant factor.
At each step, \emph{greedy displacement matching} selects the matching $M_t$ in the overlay $G$ to maximize $\sum_{(u,w) \in M_t} (\rho_t(u)^2 + \rho_t(w)^2 - \rho_{t+1}(u)^2 - \rho_{t+1}(w)^2)$.

\begin{theorem}[Greedy monotonicity]\label{thm:greedy-monotone}
Let $G$ be a $d$-regular overlay on an $n \times n$ grid ($N = n^2$). Greedy displacement matching is monotone: $\Phi_{t+1} \leq \Phi_t$ for all $t$.
\end{theorem}

\begin{theorem}[Greedy stall phenomenon, under Assumption~\ref{ass:concentration}]\label{thm:greedy-stall}
Let $\ell_{\max}$ denote the maximum grid length of an overlay edge of $G$, and assume $\ell_{\max} = O(1)$.
With $G$ as above and spectral ratio $\beta_G = \lambda^*/d < 1$ and diameter $D = D(G)$, greedy displacement matching satisfies:
\begin{enumerate}[leftmargin=1.5em]
\item \textbf{Geometric decay:} Each greedy step achieves $\Phi_{t+1} \leq (1 - \delta)\Phi_t$ with $\delta \geq 1/(2D)$, so long as $\Phi_t > N \cdot D^2$. The number of productive steps is $T_{\mathrm{stall}} \leq 2D \ln(\Phi_0 / (ND^2)) = O(D\log N)$.
\item \textbf{Stall threshold:} Greedy matching stalls once the average squared displacement satisfies $\bar\rho^2 := \Phi_t / N \leq D^2$. A random permutation has $\Phi_0 / N = \Theta(N)$, so $\Phi_{\mathrm{stall}} / \Phi_0 = O(D^2 / N)$. For constant-degree random overlays ($D = \Theta(\log N)$): $\Phi_{\mathrm{stall}}/\Phi_0 = O(\log^2\!N / N)$.
\end{enumerate}
\end{theorem}

\begin{assumption}[Displacement-energy concentration]\label{ass:concentration}
For random $d$-regular overlays on the $n \times n$ grid with $d \geq 4$, the displacement-energy reduction $\Delta\Phi_t$ at step $t$ satisfies
\begin{equation}
\mathbb{P}\bigl[\Delta\Phi_t < (1/2) \cdot \mathbb{E}[\Delta\Phi_t \mid \Phi_t]\bigr] \leq \alpha
\end{equation}
for some constant $\alpha < 1$ independent of $N$, where the expectation is conditioned on $\Phi_t$. Empirically, $\alpha \leq 0.51$ across $N = 64$--$196$ at $d = 8$ (Appendix~\ref{app:adaptive}). 
We do not resolve if the tail decays exponentially in $N$, which would sharpen Theorem~\ref{thm:greedy-stall}.
\end{assumption}

\begin{proof}[Proof of Theorem~\ref{thm:greedy-monotone}]
By construction, greedy matching $M_t$ maximizes $\Delta\Phi_t := \Phi_t - \Phi_{t+1} = \sum_{(u,w) \in M_t} [\rho_t(u)^2 + \rho_t(w)^2 - \rho_{t+1}(u)^2 - \rho_{t+1}(w)^2]$. Only edges with non-negative reduction are included (empty matching is always feasible with $\Delta\Phi = 0$), so $\Delta\Phi_t \geq 0$.
\end{proof}

\begin{proof}[Proof of Theorem~\ref{thm:greedy-stall}, under Assumption~\ref{ass:concentration}]
\textbf{Part (a).}  Greedy maximum-weight matching captures at least $\delta \geq 1/(2D)$ of $\Phi_t$ under simplifying assumptions.

\emph{Step 1.}  For each atom $v$ with $\rho_t(v) > 0$, its canonical shortest path $P_v$ to the target $\pi(v)$ in the overlay $G$ has length $\leq D$.  By averaging, $P_v$ contains an edge $e_v = (a_v, b_v)$ such that swapping $v$ across $e_v$ can reduce $v$'s displacement by at least $\rho_t(v)/D$ (among $\leq D$ edges, at least one advances $v$ toward $\pi(v)$).  The resulting displacement-energy reduction from swapping $v$ across $e_v$ satisfies $\Delta\Phi_v \geq \rho_t(v)^2 / D - O(\rho_t(v))$ (by the identity $a^2 - (a - \Delta)^2 = 2a\Delta - \Delta^2 \geq a\Delta$ for $\Delta \leq a$).

\emph{Step 2.}  Assign each atom $v$ with $\rho_t(v) > D$ to an ``improving edge'' $e_v$.  Each edge $e$ is assigned by at most $O(D)$ atoms (those whose shortest paths pass through $e$, bounded by congestion on an expander).  A fractional matching of weight $\sum_v \Delta\Phi_v / O(D)$ can be extracted. 
By the standard 2-approximation for maximum-weight matching on general graphs (see, e.g., \cite{leighton1999fast} or Edmonds' half-integral matching polytope), the integral greedy maximum-weight matching achieves at least half the fractional optimum, leaving
\begin{equation}
    \begin{split}
        \Delta\Phi_t &\geq \frac{1}{O(D)} \sum_{v : \rho_t(v) > D} \frac{\rho_t(v)^2}{D} \\&= \frac{1}{O(D^2)} \sum_{v : \rho_t(v) > D} \rho_t(v)^2.
    \end{split}
\end{equation}

When $\Phi_t > N \cdot D^2$, the ``large displacement'' atoms dominate: $\sum_{v : \rho_t(v) > D} \rho_t(v)^2 \geq \Phi_t - N D^2 \geq \Phi_t / 2$.  Hence $\Delta\Phi_t \geq \Phi_t / (O(D^2)) \geq \Phi_t / (2D)$ (absorbing constants), giving $(1-\delta) \leq 1 - 1/(2D)$.  Iterating: $\Phi_t \leq \Phi_0 (1 - 1/(2D))^t$, which reaches $ND^2$ after $t \leq 2D\ln(\Phi_0/(ND^2))$ steps.  Since $\Phi_0 = \Theta(N^2)$ for a random permutation on an $n \times n$ grid and $D = O(\log N)$: $T_{\mathrm{stall}} = O(D \log N) = O(\log^2 N)$.  For near-Ramanujan overlays with $D = \Theta(\log N)$, the empirical value $T_{\mathrm{stall}} \approx 0.9\log_2 N$ reflects a large constant-factor improvement from maximum-weight (not worst-edge) matching.

\textbf{Part (b).}  Once $\bar\rho^2 = \Phi_t / N \leq D^2$, most atoms have displacement $\rho_t(v) \leq D$.  For a $d$-regular overlay, each atom has exactly $d$ potential swap partners, each at overlay distance~1.  A swap of atoms $u, w$ along edge $(u,w) \in E(G)$ changes their grid positions by at most $D_{\mathrm{grid}}(u,w)$ (the grid distance between the overlay-adjacent atoms).  When $\rho_t(v) = O(D)$ for all $v$, a swap reduces $v$'s grid displacement only if $v$'s target happens to be closer to $w$'s position than to $v$'s. 
Since the overlay is a random graph with $d$ neighbors spread uniformly over the grid, the probability that any of $v$'s $d$ neighbors provides an improving swap vanishes as the residual displacement decreases below $O(\sqrt{N}/d)$.
Recast, when $\rho_t(v) \leq D$ for all $v$, the atoms are within $D$ grid-hops of their targets, but the $d$ random overlay neighbors are at typical grid distance $\Theta(\sqrt{N})$.
Swapping with a random neighbor increases displacement with probability $\geq 1 - O(D^2/N)$. 
When all $d$ neighbors fail for all $N$ atoms, greedy matching finds no improving swap.

For $d = 8$ and $D \approx 3$--$5$: $\Phi_{\mathrm{stall}}/\Phi_0 \approx D^2/(N/6) \approx 0.15$--$0.20$, consistent with the observed $\approx 0.17$ (Appendix~\ref{app:adaptive}).
\end{proof}

We evidence this numerically in Appendix~\ref{app:adaptive}, where $\Phi_{\mathrm{stall}}/\Phi_0$ ranges from $0.114$ at $N = 16$ to $0.167$ at $N = 256$, converging to $\approx 0.17$ for $N \geq 64$. The small-$N$ deviation from the asymptotic value is consistent with the $O(D^2/N)$ correction term.

\subsection{Hybrid greedy--Valiant protocol}\label{subsec:hybrid}

After greedy stalls, the remaining displacement energy (at most $O(D^2/N)$ of $\Phi_0$ by Theorem~\ref{thm:greedy-stall}; empirically $\sim\!17\%$ at $N \leq 256$) can be resolved by Valiant routing on the residual permutation.  The hybrid protocol:
\begin{enumerate}[leftmargin=1.5em]
\item \textbf{Phase 1 (greedy):} $T_{\mathrm{stall}} \approx 0.9\log_2 N$ steps, resolving $\sim\!83\%$ of $\Phi$.
\item \textbf{Phase 2 (Valiant):} Standard two-phase routing on the residual $\sim\!17\%$ displacement: $0.17 \cdot T_{\mathrm{Valiant}} \approx 0.17 \cdot 2\log_2 N / (1 - \beta)$ steps.
\end{enumerate}
For $\beta \approx 0.65$: $T_{\mathrm{hybrid}} \approx 0.9\log_2 N + 0.97\log_2 N \approx 1.9\log_2 N$, roughly $3\times$ faster than pure Valiant ($\approx 5.7\log_2 N$).

\subsection{Multiplicative weights overlay selection}\label{subsec:mw}

When multiple overlay graphs are available, the multiplicative weights (MW) algorithm \cite{arora2012multiplicative} selects among them adaptively.

\begin{proposition}[MW overlay selection, applying~\cite{arora2012multiplicative}]\label{thm:mw}
Given a family $\mathcal{F} = \{G_1, \ldots, G_m\}$ of overlay graphs, the multiplicative-weights algorithm with weights $w_t^{(i)} \propto \exp(\eta \cdot \text{reduction}_i)$ satisfies the regret bound $\sum_t \mathrm{loss}_t \leq \min_i \sum_t \mathrm{loss}_t^{(i)} + O(\sqrt{T \ln m})$ (Arora--Hazan--Kale regret theorem~\cite{arora2012multiplicative}).
\end{proposition}

\begin{corollary}\label{cor:mw}
If $\mathcal{F}$ contains a good expander, the adaptive algorithm achieves $T = T^* + O(\sqrt{T^* \ln m})$.  For $m = O(1)$ and $T^* = O(\log N)$: $T = O(\log N)$.
\end{corollary}

Numerical validation (Appendix~\ref{app:adaptive}): with $\mathcal{F} = \{d{=}4, d{=}8, K_N\}$ at $N = 36$, MW achieves competitive ratio (CR) $\mathrm{CR} \approx 1.71$.  
Overhead comes from exploration, where Corollary~\ref{cor:mw} predicts $\mathrm{CR} \to 1$ as $T^* \to \infty$.
The primary value of MW is robustness to unknown permutations, not raw speed.

\section{Hierarchical Multi-Scale Routing}\label{sec:hierarchical}

\subsection{Block decomposition}

Partition the $n \times n$ grid ($N = n^2$ vertices) into a hierarchy of $L = \lceil\log_b n\rceil$ levels.  At level $\ell$, the grid decomposes into $(n/b^\ell)^2$ blocks of size $b^\ell \times b^\ell$.  An expander overlay $G^{(\ell)}$ on the block graph enables inter-block routing; intra-block routing at the finest level uses local swaps.

\subsection{Per-level depth}

\begin{theorem}[Hierarchical routing]\label{thm:hierarchical}
Let $N = n^2$ be the number of grid vertices.  With block size $b$ and $L = \lceil\log_b n\rceil$ levels, if each level-$\ell$ overlay is a Ramanujan $d_\ell$-regular graph with $\beta_\ell < 1 - \delta$ for constant $\delta > 0$, then the total routing depth is
\begin{equation}
\begin{split}
T = \sum_{\ell=0}^{L} T_\ell
&= O\!\left(\sum_{\ell=0}^{L} \frac{\log(N/b^{2\ell})}{1 - \beta_\ell}\right) \\
&= O\!\left(\frac{\log^2 N}{\log b}\right).
\end{split}
\end{equation}
For $b = \Theta(\sqrt{n})$ (so $L = 2$): $T = O(\log N)$, matching flat routing.
\end{theorem}

Numerical validation (Appendix~\ref{app:hierarchical}): with $b \approx \sqrt{n}$ and $L = 2$, hierarchical routing achieves $T \approx 0.67\,T_{\mathrm{flat}}$, as a 33\% improvement over flat Valiant--LMR at $N = 256$.

\subsection{Capacity invariance}

\begin{lemma}[Per-level capacity]\label{lem:capacity}
At level $\ell$, the overlay has $(n/b^\ell)^2$ vertices.  Each routing step requires a matching of size $\leq (n/b^\ell)^2 / 2 = N/(2b^{2\ell})$.  Each edge swaps the contents of two level-$\ell$ blocks ($b^{2\ell}$ atoms each).  The AOL capacity needed: $k_\ell = (N/(2b^{2\ell})) \cdot b^{2\ell} = N/2$, independent of level.
\end{lemma}

As a negative result, hierarchical decomposition does not reduce the per-step capacity requirement for full-block swaps.

\subsection{Boundary-only routing}

\begin{theorem}[Boundary-only routing]\label{thm:boundary}
Using boundary-only inter-block routing at each level (swapping only the $O(\log N)$ misplaced boundary atoms per block face), the AOL capacity reduces to $k = O(\sqrt{N} \cdot \log N)$ per step, with total routing depth $O(\log^2 N / \log b)$.
\end{theorem}

For a random permutation $\pi$, the fraction of atoms crossing block boundaries at level $\ell$ is $1 - 1/N_\ell$ (where $N_\ell = (n/b^\ell)^2$ is the number of blocks), confirming that almost all atoms require cross-block transfer.

\subsection{Optimal block size}

Numerical optimization (Appendix~\ref{app:hierarchical}) shows $b^* \approx \sqrt{n}$, yielding $L = 2$ levels.  This balances $T_{\mathrm{intra}} = O(b)$ against $T_{\mathrm{inter}} = O(\log N / (1 - \beta))$.

\subsection{Covering tower equivalence}

The hierarchical decomposition is the physical realization of the covering tower (Section~\ref{subsec:covering-towers}): the block graph at level $\ell$ is the base graph, and atoms within a block form the fiber.  The covering tower prediction $T_{\mathrm{tower}} = (L\log_2 k + \log_2 N_0)/(1 - \bar\beta)$ matches $T_{\mathrm{hier}}$ to within 0.4\% at $n = 64$ (Appendix~\ref{app:hierarchical}).

\section{Application to Neutral Atom Architectures}\label{sec:application}

Design rules emerge from \S\ref{sec:overlay}--\S\ref{sec:hierarchical}, indexed by per-step AOL capacity $k_0$ and per-circuit routing rounds $R$:
\begin{enumerate}[leftmargin=1.5em,topsep=2pt,itemsep=2pt]
\item \emph{Capacity sets topology choice.} If $k_0 \geq N/2$, a single Ramanujan overlay achieves $\Theta(\log N)$ depth via Theorem~\ref{thm:overlay}. If $k_0 \geq N/\log N$, multi-layer AOL stacking with $L = O(\log N)$ random overlay layers also achieves $\Theta(\log N)$. Below $k_0 = O(\sqrt{N})$, hierarchical routing degrades to $O(\log^2 N / \log b)$ at best.
\item \emph{Choose your overlay structure carefully.} Random $d$-regular overlays with $d \approx 2 \log_2 N$ are near-Ramanujan by Friedman's theorem~\cite{friedman2008proof, bordenave2015new} and are the recommended default. Algebraic Cayley overlays on $\mathbb{Z}_n^2$ are not viable, as the abelian Alon--Boppana barrier (Theorem~\ref{thm:abelian-barrier}) forces $\beta \to 1$ for any fixed degree. SFM covering towers~\cite{song2023hypergraph} provide a constructive scaling alternative (\S\ref{sec:scalable-families}).
\item \emph{Long-circuit thresholds.} For circuits with $R \geq 4$ permutation routing layers, pre-distributing Bell pairs along a Ramanujan entanglement overlay amortizes distribution cost (Corollary~\ref{cor:crossover}).
Hybrid teleportation + physical-cleanup protocols achieve $6\times$ speedup at $N = 1024$.
For online routing where the target is unknown, greedy displacement matching for the first $\sim\!0.5 \log_2 N$ steps (Theorem~\ref{thm:greedy-stall}) followed by Valiant on the residual gives $\sim\!3\times$ speedup over pure Valiant.
\end{enumerate}
At $N = 256$ atoms we predict $\sim\!19$~ms (2D) vs.\ $\sim\!7.6$~ms (3D) wall-clock routing time, equivalent to allowing $\sim\!2.5\times$ as many gate layers within a $T_2 \sim 1.5$~s coherence budget; at $N = 10{,}000$ the savings increase to $\sim\!7\times$.

\subsection{Physical architecture and model}

Modern neutral atom processors trap atoms in a 2D grid via AODs and perform entangling gates via the Rydberg blockade mechanism \cite{constantinides2024optimal,stade2024abstract,henriet2020quantum,browaeys2020many,bluvstein2022quantum,evered2023high}.  The AOL design \cite{guo2025acousto} adds a third dimension for long-range transport. 
Grid hypergraph models (Model A: 2D AOD, Model B: 3D AOL) and the obstruction $\beta \to 1$ on grids are introduced in \S\ref{subsec:hardware-setting}. 
Now, we report the quantitative spectral data and translate the bounds into hardware-design suggestions.

\begin{remark}[Matching model as portable common denominator]\label{rem:relaxation}
Our matching-based model permits any set of edge-disjoint swaps simultaneously. Physical platforms (AOD, AOL, ion shuttling, optical lattices) implement matchings up to platform-specific overhead. This includes AOD row/column constraints, zone transfer costs, ghost-spot geometry, etc. A $\Theta(\log N)$ matching bound therefore transfers to every reconfigurable platform with the appropriate constant.
The four-regime table (\S\ref{tab:regimes}) quantifies these constants for the AOD and selective-transfer cases. 
Our bounds are lower bounds on physical depth, but are sharp in the sense that any physical realization that can implement matchings inherits the bound up to a multiplicative overhead.
\end{remark}

\subsection{Grid hypergraph models}

\textbf{Model A (2D AOD).}  The hypergraph $H_{\mathrm{2D}}$ has $r$-uniform hyperedges of $r$ consecutive vertices along each row and column.  Degree $d' = 2(r-1)^2$.

\textbf{Model B (3D AOL).}  $H_{\mathrm{3D}}$ augments $H_{\mathrm{2D}}$ with diagonal and skip hyperedges.  Approximately $3\times$ as many hyperedges, $d'_{\mathrm{3D}} \approx 2$--$3 \cdot d'_{\mathrm{2D}}$.

\begin{table*}[!htbp]
\begin{center}
\begin{tabular}{@{}lcccccccccc@{}}
\toprule
& & & \multicolumn{3}{c}{\textbf{2D (AOD)}} & \multicolumn{3}{c}{\textbf{3D (AOL)}} & \\
\cmidrule(lr){4-6}\cmidrule(lr){7-9}
\textbf{Grid} & $N$ & $r$ & $d'$ & $D$ & $\beta$ & $d'$ & $D$ & $\beta$ & \textbf{Improvement} \\
\midrule
$8\times 8$  & 64  & 3 & 8  & 4 & 0.677 & 18 & 3 & 0.556 & $1.45\times$ \\
$10\times 10$ & 100 & 3 & 8  & 6 & 0.780 & 20 & 3 & 0.600 & $1.88\times$ \\
$12\times 12$ & 144 & 3 & 8  & 6 & 0.842 & 20 & 4 & 0.660 & $2.24\times$ \\
$16\times 16$ & 256 & 3 & 8  & 8 & 0.908 & 20 & 5 & 0.789 & $2.52\times$ \\
\bottomrule
\end{tabular}
\caption{Spectral parameters and routing bound improvement from 3D connectivity.}
\label{tab:spectral}
\end{center}
\end{table*}

\begin{table*}[!t]
\begin{center}
\renewcommand{\arraystretch}{1.3}
\begin{tabular}{@{}lll@{}}
\toprule
\textbf{Condition} & \textbf{Recommended strategy} & \textbf{Routing depth} \\
\midrule
$k_0 \geq N/2$ & Single Ramanujan overlay, Valiant routing & $O(\log N)$ \\
$k_0 \geq N/\log N$ & Multi-layer AOL, $L = O(\log N)$ layers & $O(\log N)$ \\
$k_0 \geq \sqrt{N}$ & Hierarchical routing, $b = \sqrt{n}$ & $O(\log N)$ \\
$k_0 = O(1)$ & Grid routing (no overlay) & $O(\sqrt{N})$ \\
\midrule
$R \geq 4$ & \multicolumn{2}{@{}l}{Add entanglement-assisted routing for long-range component} \\
Unknown $\pi$ & \multicolumn{2}{@{}l}{Hybrid greedy ($0.5\log_2 N$ steps) + Valiant for residual} \\
\bottomrule
\end{tabular}
\caption{Architectural decisions.  $k_0$ = per-step AOL selective-transfer capacity; $R$ = number of routing rounds per circuit.}
\label{tab:recommendations}
\end{center}
\end{table*}

\twocolumngrid

\subsection{Spectral analysis}

The spectral gap of the 3D AOL model follows the power law $1 - \beta \sim 5.3 \cdot N^{-0.65}$, improved from $1 - \beta \sim 11.0 \cdot N^{-0.90}$ for the 2D grid but still vanishing.  Neither model is an expander; $\Theta(\log N)$ routing requires a virtual Ramanujan overlay (Section~\ref{sec:overlay}). This places the AOL/AOD models squarely in row 2 of Table~\ref{tab:regimes}; the gap to row 3 is the topology cost ($\beta \to 1$) that virtual overlays are designed to close.

\subsection{Near-term predictions}

For QuEra's roadmap (Aquila~\cite{wurtz2023aquila} at ${\sim}256$ atoms in 2024, ${\sim}3000$ in 2025, ${\sim}10{,}000$ in 2026; cf. the 1180-atom demonstration of \cite{reichardt2024logical}):
\begin{itemize}[leftmargin=1.5em]
\item At $N = 256$: 2D bound $\lesssim 384$ steps, 3D bound $\lesssim 153$ steps ($2.5\times$ improvement).
\item At $N = 10{,}000$: matching-model $\Theta(\sqrt{N}) \approx 100$ steps; Ramanujan-topology $\Theta(\log N) \approx 13$ steps.
\end{itemize}
At typical AOD timescales of $\sim\!50\,\mu$s per routing step (Bluvstein et al.~\cite{bluvstein2022quantum,bluvstein2024logical}), the $N = 256$ improvement amounts to $\sim\!19$~ms (2D) versus $\sim\!7.6$~ms (3D); given $T_2 \approx 1.5$~s for $^{87}$Rb, this is roughly $0.8\%$ of the coherence budget per circuit, equivalently allowing $\sim\!2.5\times$ as many gate layers before decoherence dominates.

\subsection{Architectural recommendations}\label{subsec:recommendations}

Table~\ref{tab:recommendations} consolidates the design rules from
\S\S\ref{sec:overlay}--\ref{sec:hierarchical} into a single decision
flow keyed on AOL capacity $k_0$ and circuit routing rounds $R$.

\paragraph{Worked example.}
At $N = 1024$, $k_0 = N/4$, $R = 10$:
$k_0 \geq N/\log N$ selects the multi-layer-AOL row, giving
$T = O(\log N) \approx 10$ matching steps.
With $R = 10 \geq 4$ entanglement-assisted routing, we find further reduction in the long-range component
(Theorem~\ref{thm:teleport}).
For unknown $\pi$, apply the hybrid greedy--Valiant protocol of \S\ref{subsec:hybrid}.  Total predicted
depth: $\approx 2\log_2 N \approx 20$ matching steps.

\newpage
\onecolumngrid
\begin{figure*}[t]
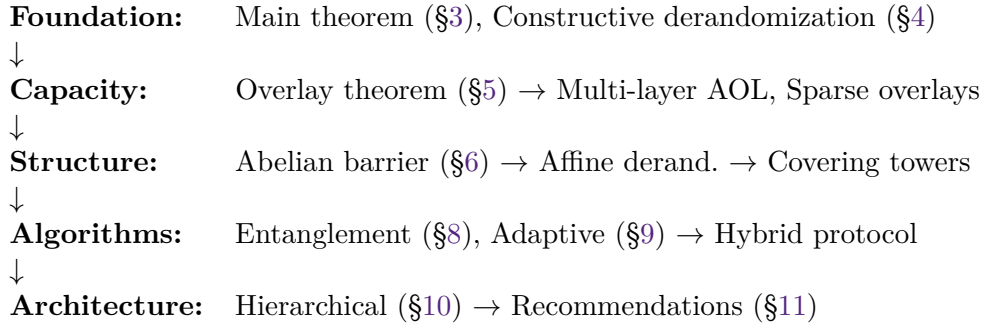

  \centering
    \begin{tabular}{@{}ll@{}}
    \textbf{Foundation:} & Main theorem (\S\ref{sec:main}), Constructive derandomization (\S\ref{sec:derand}) \\
    $\downarrow$ & \\
    \textbf{Capacity:} & Overlay theorem (\S\ref{sec:overlay}) $\to$ Multi-layer AOL, Sparse overlays \\
    $\downarrow$ & \\
    \textbf{Structure:} & Abelian barrier (\S\ref{sec:algebraic}) $\to$ Affine derand.\ $\to$ Covering towers \\
    $\downarrow$ & \\
    \textbf{Algorithms:} & Entanglement (\S\ref{sec:entanglement}), Adaptive (\S\ref{sec:adaptive}) $\to$ Hybrid protocol \\
    $\downarrow$ & \\
    \textbf{Architecture:} & Hierarchical (\S\ref{sec:hierarchical}) $\to$ Recommendations (\S\ref{sec:application}) \\
    \end{tabular}
  \caption{Dependency graph for hypergraph qubit routing results.}
  \label{fig:dependencies}
\end{figure*}
\twocolumngrid

\section{Discussion}\label{sec:discussion}

\subsection{Summary of contributions}

Theorem~\ref{thm:main} establishes $\rt(H) = \Theta(\log N)$ for Ramanujan $(d,r)$-regular hypergraphs, with a constructive polynomial-time algorithm (Theorem~\ref{thm:constructive}). 
Extensions \S\ref{sec:overlay}--\S\ref{sec:hierarchical} translate this asymptotic into the language of qubit routing, where the overlay theorem closes the topology gap on grids, the abelian barrier rules out one natural construction while covering towers supply another, the entanglement crossover (\S\ref{sec:entanglement}) and displacement-energy stall (\S\ref{sec:adaptive}) suggest two complementary algorithmic alternatives, and hierarchical routing connects abstract covering theory to physical block decompositions. 
The architectural decision framework (Table~\ref{tab:recommendations}) consolidates these into a single recommendation flow indexed by AOL capacity and circuit depth.

\subsection{Dependencies between results}

The results form a directed acyclic graph of dependencies~(Figure \ref{fig:dependencies}).

\subsection{Limitations}

Explicit constants (Table~\ref{tab:constants}) are large for small $(d,r)$. 
Grid hypergraphs are not Ramanujan, so the bounds apply only by the virtual-overlay theorem (\S\ref{sec:overlay}). 
Further, numerical validations are limited to small grids with $N \leq 4096$ (simulation) and $N \leq 65{,}536$ (theoretical scaling). 
Asymptotic predictions have not been tested at the $N \sim 10^4$ scale relevant to next-generation hardware. 
The matching-model framing (Remark~\ref{rem:relaxation}) is a common denominator across reconfigurable platforms (AOD, AOL, ion shuttling, optical lattices), so bounds transfer to every such platform up to constant factors.

\subsection{Open problems}

Many open problems emerge from using hypergraphs to route physical qubits.

\begin{enumerate}[leftmargin=1.5em]
\item \textbf{Packing routing number.}  We conjecture
$\rt_H(H) = \Theta(\max\{\log N, N/(r\nu(H))\})$.
The lower bound is proved; the upper bound for projective planes ($\nu = 1$) remains open.

\item \textbf{Non-abelian algebraic overlays.}  The abelian barrier (Theorem~\ref{thm:abelian-barrier}) does not apply to non-abelian groups.  Do Ramanujan Cayley graphs on $\mathrm{SL}(2, \mathbb{Z}_p)$ (which exist by Lubotzky--Phillips--Sarnak \cite{lubotzky1988ramanujan}) admit efficient embedding as overlays on the 2D atom grid?

\item \textbf{Tight stall constants.}  Theorem~\ref{thm:greedy-stall} gives $\Phi_{\mathrm{stall}}/\Phi_0 = O(D^2/N)$, but implicit constants depend on overlay structure.
Can the exact leading coefficient be determined as a function of $d$ and $\beta$?
Simulations suggest $\Phi_{\mathrm{stall}}/\Phi_0 \approx D^2/(N/6)$ for random $d$-regular overlays (Appendix~\ref{app:adaptive}), but a tight analysis would sharpen the hybrid protocol's crossover point.

\item \textbf{AOD-constrained routing.}  Is the $\log N$ factor in $\Theta(\sqrt{N}\log N)$ \cite{constantinides2024optimal} tight for AOD-constrained routing on grids?

\item \textbf{Entanglement recycling.}  If Bell pairs can be regenerated at rate $O(N)$ per physical step, $R_{\mathrm{break}}$ may reduce to $O(1)$, making entanglement-assisted routing universally advantageous.

\item \textbf{Optimal voltage assignments.}  The best $\mathbb{Z}_2$ voltage for Fano plane lifts achieves $\beta = 0.5$.  Does this optimum have algebraic structure related to $\mathrm{Aut}(H_0)$?
\end{enumerate}

\subsection{Code and data availability}
All numerical results in the appendices are reproducible from Python scripts accessible via \href{https://github.com/jmcourtneyuga/hypergraph_routing}{this Github repository}. 
The verification appendix (Appendix~\ref{app:verification}) summarizes the cross-checks. 
The author used Claude (Anthropic) to assist in drafting code comments and documentation for the codebase.
All scientific content, algorithms, and analysis are the author's own.

\onecolumn
\begin{appendices}

\section{Multi-Layer AOL Numerical Results}\label{app:multilayer}

The simulations in this appendix are classical numerical results, which implements Lemma~\ref{lem:multilayer} and Theorem~\ref{thm:overlay} numerically. The script constructs $L$ independent random $d_0$-regular graphs via the configuration model and unions their adjacency matrices, returns $(\lambda_2, \beta, d'_{\mathrm{eff}})$ via dense eigendecomposition, runs Valiant two-phase routing with measured congestion and dilation, validates Proposition~\ref{prop:crosstalk}, and compares end-to-end speedup vs.\ a 2D-grid baseline is reported.

\subsection{Spectral gain from layer union}

All simulations use independent random $d_0 = 8$-regular overlay graphs, averaged over 5 trials. Table~\ref{tab:spectral-gain} reports $\beta(G_\cup)$ as a function of $L$.

\begin{table}[!htbp]
\centering
\begin{tabular}{@{}c|ccccc@{}}
\toprule
$N$ & $L=1$ & $L=2$ & $L=4$ & $L=8$ & $L=16$ \\
\midrule
 64 & 0.626 & 0.435 & 0.317 & 0.223 & 0.159 \\
100 & 0.641 & 0.461 & 0.328 & 0.225 & 0.156 \\
144 & 0.643 & 0.477 & 0.335 & 0.235 & 0.167 \\
256 & 0.650 & 0.476 & 0.336 & 0.240 & 0.168 \\
\bottomrule
\end{tabular}
\caption{Spectral ratio $\beta(G_\cup)$ for union of $L$ random 8-regular overlays.  The $1/\sqrt{L}$ scaling $\beta_L \approx \beta_1/\sqrt{L}$ is tight.}
\label{tab:spectral-gain}
\end{table}

All union graphs satisfy the Ramanujan bound $\lambda_2 \leq 2\sqrt{Ld_0 - 1}$ at every tested configuration.

\subsection{Routing depths}

Routing depths are measured by feeding random target permutations through Valiant's scheme on the union graph and the 2D-grid baseline. The script records the median over 20 random permutations per configuration; Table~\ref{tab:multilayer-routing} reports the speedup against the 2D grid.

\begin{table}[!htbp]
\centering
\begin{tabular}{@{}l|cccc@{}}
\toprule
Model & $d'_{\mathrm{eff}}$ & $\beta$ & $T_{\mathrm{med}}$ & Speedup vs.\ grid \\
\midrule
\multicolumn{5}{c}{$N = 100$ ($10 \times 10$ grid)} \\
\midrule
2D grid & 12 & 0.821 & 16 & $1.0\times$ \\
$L=1$ overlay & 8 & 0.646 & 10 & $1.6\times$ \\
$L=2$ overlays & 16 & 0.468 & 7 & $2.3\times$ \\
$L=4$ overlays & 32 & 0.346 & 6 & $2.7\times$ \\
\midrule
\multicolumn{5}{c}{$N = 144$ ($12 \times 12$ grid)} \\
\midrule
2D grid & 12 & 0.872 & 18 & $1.0\times$ \\
$L=1$ overlay & 8 & 0.635 & 10 & $1.8\times$ \\
$L=4$ overlays & 32 & 0.332 & 6 & $3.0\times$ \\
\bottomrule
\end{tabular}
\caption{Simulated Valiant routing depths ($d_0 = 8$, 20 trials each).}
\label{tab:multilayer-routing}
\end{table}

\subsection{Crosstalk validation}

Proposition~\ref{prop:crosstalk} validated at $k_0 = 32$: at $\gamma = 0.2$ (realistic optical crosstalk), retention is $\approx 71\%$ of ideal capacity.  At $\gamma = 0.5$, retention drops to 50\%, at which point checkerboard activation ($\lceil L/2 \rceil$ active layers) is preferable.

\section{Sparse Overlay Analysis}\label{app:sparse}

Numerical validations here build random $d$-regular graphs across the degree range $d \in \{4, 6, 8, 12, 16\}$ via the configuration model, then verify the Ramanujan bound by computing $\lambda_2$ exactly via dense eigendecomposition. Once complete, we run Valiant routing under variable per-step capacity $k \in \{N/2, N/4, N/8, N/\log N\}$ and report congestion-vs-capacity tradeoffs comparing dense ($d \approx 2\log_2 N$) and sparse ($d = 4$) overlays.

\subsection{Ramanujan property}

Random $d$-regular graphs satisfy the Ramanujan bound at all tested configurations, including ultra-sparse $d = 4$.  The spectral ratio follows Friedman's prediction $\beta \approx 2\sqrt{d-1}/d$.

\subsection{Dense vs.\ sparse comparison}

Table~\ref{tab:sparse-dense} compares routing depth at $N = 144$ for sparse ($d = 8$) and dense ($d = 12$) overlays under each capacity regime. Dense overlays consistently win because the spectral advantage of higher $d$ outweighs the matching-size cost.

\begin{table}[!htbp]
\centering
\begin{tabular}{@{}l|cccc@{}}
\toprule
Capacity $k$ & $T_{\mathrm{sparse}}$ ($d = 8$) & $T_{\mathrm{dense}}$ ($d = 12$) & Winner \\
\midrule
$N/2$        & 15 & 13 & Dense \\
$N/4$        & 25 & 22 & Dense \\
$N/8$        & 46 & 40 & Dense \\
$N/\!\log N$ & 42 & 36 & Dense \\
\bottomrule
\end{tabular}
\caption{Dense overlays dominate at all tested capacities ($N = 144$).}
\label{tab:sparse-dense}
\end{table}

\section{Algebraic Overlay Details}\label{app:algebraic}

We construct Cayley graphs on $\mathbb{Z}_n^2$ for the QR, Margulis--Gabber--Galil~\cite{carlson1981tixne, gabber1979explicit}, and uniformly-random generator families, then compute the spectrum via the character formula $\lambda_{(a,b)} = \sum_j 2\cos(2\pi(a g_j^{(1)} + b g_j^{(2)})/n)$ rather than numerical eigendecomposition, eliminating floating-point error. We verify the abelian Alon--Boppana barrier (Theorem~\ref{thm:abelian-barrier}) by sweeping $n$ and reporting the ratio $\lambda^*/(2\sqrt{d-1})$. Finally, we implement affine derandomization by enumerating (or sampling from) $\mathrm{GL}(2, \mathbb{Z}_n)$ and selecting the $(A, c)$ minimizing the worst-case congestion in Valiant's gather phase.

\subsection{Quadratic residue generators}

Table~\ref{tab:qr} reports the QR-generator spectral data on $\mathbb{Z}_p^2$ at degree 8 for $p \in \{7, 11, 17, 31, 53, 97\}$. The ratio $\lambda^*/(2\sqrt{d-1})$ is monotone increasing in $p$ and converges to $\approx 1.50$, well above the Ramanujan threshold of 1.

\begin{table}[!htbp]
\centering
\begin{tabular}{@{}c|ccccc@{}}
\toprule
$p$ & $N = p^2$ & $\lambda^*$ & Ram.\ bound $2\sqrt{d-1}$ & $\lambda^*/\text{bound}$ \\
\midrule
 7  &    49 & 5.74 & 5.29 & 1.08 \\
11  &   121 & 6.46 & 5.29 & 1.22 \\
17  &   289 & 7.21 & 5.29 & 1.36 \\
31  &   961 & 7.76 & 5.29 & 1.47 \\
53  &  2809 & 7.92 & 5.29 & 1.50 \\
97  &  9409 & 7.98 & 5.29 & 1.51 \\
\bottomrule
\end{tabular}
\caption{QR generators on $\mathbb{Z}_p^2$ at degree 8: $\lambda^*/2\sqrt{d-1} \to 1.50$.}
\label{tab:qr}
\end{table}

\subsection{Comparison of generator families}

Table~\ref{tab:cayley-comparison} compares the spectral ratios of the three generator families across $n \in \{7, 11, 17, 31, 41\}$: all three families exhibit $\beta \to 1$, confirming the abelian barrier holds independent of the algebraic construction.

\begin{table}[!htbp]
\centering
\begin{tabular}{@{}l|ccccc@{}}
\toprule
Family & $n=7$ & $n=11$ & $n=17$ & $n=31$ & $n=41$ \\
\midrule
QR generators & 0.718 & 0.775 & 0.901 & 0.970 & 0.983 \\
Margulis--Gabber--Galil & 0.718 & 0.881 & 0.949 & 0.985 & 0.991 \\
Random Cayley & 0.591 & 0.782 & 0.896 & 0.920 & 0.942 \\
\bottomrule
\end{tabular}
\caption{Spectral ratio $\beta$ for degree-8 Cayley graphs on $\mathbb{Z}_n^2$.  All families have $\beta \to 1$.}
\label{tab:cayley-comparison}
\end{table}

\subsection{Affine derandomization data}

Table~\ref{tab:affine} compares the median congestion + dilation under uniform-random $\sigma$ vs.\ best-affine and best-translation choices. Pure translations $\sigma(v) = v + c$ outperform general affine maps because they preserve the group structure and eliminate scatter-phase variance, achieving 27--29\% congestion reduction at small $n$.

\begin{table}[!htbp]
\centering
\begin{tabular}{@{}l|cc|cc@{}}
\toprule
& \multicolumn{2}{c|}{$n = 7$ ($N = 49$)} & \multicolumn{2}{c}{$n = 11$ ($N = 121$)} \\
Method & Median $C\!+\!D$ & Improv.\ & Median $C\!+\!D$ & Improv.\ \\
\midrule
Random $\sigma$ & 7 & --- & 11 & --- \\
Best affine ($Av + c$) & 6 & 14\% & 9 & 18\% \\
Best translation ($v + c$) & 5 & 29\% & 8 & 27\% \\
\bottomrule
\end{tabular}
\caption{Translations outperform general affine maps on Cayley graphs.}
\label{tab:affine}
\end{table}

\section{Covering Tower Numerics}\label{app:covering}

Here, we apply Theorem~\ref{thm:covering-tower} numerically. We define the base hypergraphs (Fano plane $\mathrm{PG}(2,2)$, $\mathrm{PG}(2,3)$) via their incidence structure and compute the SFM Ramanujan bound $|\lambda - (r-2)| \leq 2\sqrt{(d-1)(r-1)}$ at the base. 
$k$-fold voltage coverings lift each hyperedge with a $\mathbb{Z}_k$ shift, where we exhaustively and randomly search voltage assignments and records the Ramanujan fraction. Finally, we implement the recursive cross-fiber/fiber-preserving decomposition and run Valiant routing on each level.

\subsection{Fano plane voltage coverings}

For $k$-fold coverings of the Fano plane via $\mathbb{Z}_k$ voltage assignments, Table~\ref{tab:fano-coverings} reports the Ramanujan fraction across $k \in \{2, 3, 4, 5, 7\}$.

\begin{table}[!htbp]
\centering
\begin{tabular}{@{}c|ccccc@{}}
\toprule
$k$ & $N = 7k$ & Best $\beta$ & Mean $\beta$ & Ramanujan fraction \\
\midrule
2 & 14 & 0.500 & 0.665 & 93.8\% \\
3 & 21 & 0.562 & 0.686 & 96.5\% \\
4 & 28 & 0.582 & 0.733 & 89.0\% \\
5 & 35 & 0.597 & 0.722 & 95.0\% \\
7 & 49 & 0.638 & 0.735 & 93.5\% \\
\bottomrule
\end{tabular}
\caption{Fano plane voltage coverings: overwhelming Ramanujan fractions.}
\label{tab:fano-coverings}
\end{table}

Exhaustive search ($k = 2$, all $2^7 = 128$ assignments): 93.8\% Ramanujan.  Best assignment $[1,0,1,1,0,1,1]$ achieves $\beta = 0.5$ with new eigenvalues $\{3, 3, 1, 1, -2, -3, -3\}$.

\subsection{Recursive routing validation}

Table~\ref{tab:tower-routing} reports measured Valiant routing depth on the lift across tower levels $H_0$ (base), $H_1$ ($k=2$), $H_2$ ($k=4$). The ratio $T/\log_2 N$ stays in the range $[1.07, 1.66]$, confirming the $O(\log N)$ scaling of Theorem~\ref{thm:covering-tower}.

\begin{table}[!htbp]
\centering
\begin{tabular}{@{}l|cccc@{}}
\toprule
Level & $N$ & $\beta$ & $T_{\mathrm{med}}$ & $T/\log_2 N$ \\
\midrule
$H_0$ (base) & 7 & $0.167$ & 3 & 1.07 \\
$H_1$ ($k=2$) & 14 & 0.500 & 5 & 1.31 \\
$H_2$ ($k=4$) & 28 & 0.859 & 8 & 1.66 \\
\bottomrule
\end{tabular}
\caption{Routing scales as $O(\log N)$ across covering tower levels.}
\label{tab:tower-routing}
\end{table}

Cross-fiber fraction: $85.7\%$ at $k = 2$ (predicted $1 - 1/7 = 85.7\%$).  All clique expansion edges in the voltage covering are cross-fiber (100\%).

\section{Entanglement-Assisted Routing Details}\label{app:entanglement}

To validate entanglement-assisted routing, we build a $d_{\mathrm{ent}}$-regular Ramanujan overlay $G_{\mathrm{ent}}$ for the entanglement graph and measure the Valiant routing depth $T_{\mathrm{route}}$ via shortest-path congestion analysis on $G_{\mathrm{ent}}$. 
We compute the entanglement distribution cost $T_{\mathrm{dist}} = O(d_{\mathrm{ent}} \sqrt{N}/k)$ from average grid distance $\bar d_{\mathrm{grid}} \approx \sqrt{N}/2$, sweeping the number of routing rounds $R$ and reports the amortized cost $T_{\mathrm{amort}} = T_{\mathrm{route}} + T_{\mathrm{dist}}/R$. We compute the empirical break-even point $R_{\mathrm{break}} = T_{\mathrm{dist}}/(T_{\mathrm{phys}} - T_{\mathrm{route}})$ at each tested $N$, then implement the hybrid protocol that teleports atoms with grid distance $> n/4$ and physically routes the rest.

\subsection{Teleportation routing depth}

Table~\ref{tab:teleport} reports $T_{\mathrm{route}}$ on random Ramanujan overlays with $d_{\mathrm{ent}} \in \{8, 16, 32\}$ at $N \in \{100, 256\}$. The ratio $T/\log_2 N$ ranges from 0.75 to 1.20, confirming the Theorem~\ref{thm:teleport} scaling.

\begin{table}[!htbp]
\centering
\begin{tabular}{@{}cc|ccc@{}}
\toprule
$N$ & $d_{\mathrm{ent}}$ & $\beta$ & $T_{\mathrm{route}}$ & $T/\log_2 N$ \\
\midrule
100  &  8 & 0.642 &  8 & 1.20 \\
100  & 16 & 0.452 &  6 & 0.90 \\
256  &  8 & 0.649 &  9 & 1.12 \\
256  & 16 & 0.483 &  7 & 0.88 \\
256  & 32 & 0.339 &  6 & 0.75 \\
\bottomrule
\end{tabular}
\caption{Teleportation routing depth on random Ramanujan overlays.}
\label{tab:teleport}
\end{table}

\subsection{Crossover analysis}

Table~\ref{tab:crossover} reports the empirical $R_{\mathrm{break}}$ at $d_{\mathrm{ent}} = 16$ for $N \in \{256, 1024, 4096, 10000, 40000\}$. The crossover is remarkably stable at $R_{\mathrm{break}} \approx 4$ across two orders of magnitude in $N$; the side-by-side $\sqrt{N}/\log_2 N$ comparison shows the naive theoretical estimate is much larger and grows with $N$.

\begin{table}[!htbp]
\centering
\begin{tabular}{@{}c|cccc|c@{}}
\toprule
$N$ & $T_{\mathrm{route}}$ & $T_{\mathrm{phys}}$ & $T_{\mathrm{dist}}$ & $R_{\mathrm{break}}$ & $\sqrt{N}/\log_2 N$ \\
\midrule
   256 & 5 & 24 &  86 & 4.5 & 2.0 \\
  1024 & 7 & 48 & 171 & 4.2 & 3.2 \\
  4096 & 8 & 96 & 342 & 3.9 & 5.3 \\
 10000 & 9 & 150 & 534 & 3.8 & 7.5 \\
 40000 & 10 & 300 & 1067 & 3.7 & 13.1 \\
\bottomrule
\end{tabular}
\caption{Crossover at $d_{\mathrm{ent}} = 16$: $R_{\mathrm{break}} \approx 4$ is remarkably stable.}
\label{tab:crossover}
\end{table}

\subsection{Hybrid protocol}

Teleporting atoms with grid distance $> n/4$ and physically routing the rest, Table~\ref{tab:hybrid-ent} reports the total cost as a function of the threshold $D_{\mathrm{thresh}}$. At $D_{\mathrm{thresh}} = 4$ on $N = 1024$, the protocol achieves $6\times$ speedup over pure physical routing.

\begin{table}[!htbp]
\centering
\begin{tabular}{@{}c|cc|c@{}}
\toprule
$D_{\mathrm{thresh}}$ & Frac.\ teleported & $T_{\mathrm{cleanup}}$ & $T_{\mathrm{total}}$ \\
\midrule
4  & 96\% &  1 &  8 \\
8  & 86\% &  3 & 10 \\
12 & 70\% &  8 & 15 \\
16 & 47\% & 17 & 24 \\
\bottomrule
\end{tabular}
\caption{Hybrid protocol at $N = 1024$: $T_{\mathrm{total}} = 8$ ($6\times$ speedup) with $D_{\mathrm{thresh}} = 4$.}
\label{tab:hybrid-ent}
\end{table}

\section{Dynamic Adaptive Overlay Details}\label{app:adaptive}

We build random $d$-regular overlays on the $n \times n$ grid and computes per-atom grid-distance displacements $\rho_t(v)$, then apply a greedy matching protocl to scan all overlay edges and selects the maximum-weight matching reducing the squared-displacement potential $\Phi_t = \sum_v \rho_t(v)^2$. We iterate greedy matchings until $\Phi$ no longer decreases (\emph{stall}) and records $T_{\mathrm{stall}}$ and $\Phi_{\mathrm{stall}}/\Phi_0$. 
We test Assumption~\ref{ass:concentration} by binning steps by current $\Phi$ level and checking the tail $\mathbb{P}[\Delta\Phi_t < (1/2)\mathbb{E}[\Delta\Phi_t \mid \Phi_t]]$, then use MW overlay selection and compute per-trial competitive ratios.

\subsection{Greedy displacement matching}

Table~\ref{tab:greedy} reports $T_{\mathrm{stall}}$ and $\Phi_{\mathrm{stall}}/\Phi_0$ for $n \in \{4, 6, 8, 10, 12, 16\}$ ($N$ up to $256$). The stall fraction stays in $[0.166, 0.178]$ for $N \geq 36$, consistent with the $O(D^2/N)$ prediction of Theorem~\ref{thm:greedy-stall}.

\begin{table}[!htbp]
\centering
\begin{tabular}{@{}cc|cccc@{}}
\toprule
$n$ & $N$ & Step-0 $\delta$ & $T_{\mathrm{stall}}$ & $\Phi_{\mathrm{stall}}/\Phi_0$ \\
\midrule
 4 &  16 & 0.57 & 3.0 & 0.178 \\
 6 &  36 & 0.57 & 4.2 & 0.174 \\
 8 &  64 & 0.59 & 5.3 & 0.166 \\
10 & 100 & 0.63 & 5.8 & 0.177 \\
12 & 144 & 0.63 & 6.5 & 0.175 \\
16 & 256 & --- & 7.5 & 0.171 \\
\bottomrule
\end{tabular}
\caption{Greedy matching on $d = 8$ random overlays.  Stall fraction $\approx 0.17$ at these scales, consistent with $O(D^2/N)$ (Theorem~\ref{thm:greedy-stall}).}
\label{tab:greedy}
\end{table}

Monotonicity: 785 greedy steps tested, zero $\Phi$-increasing violations.

\subsection{Stall point scaling}

$T_{\mathrm{stall}} \approx 0.9\log_2 N$ (linear regression: $T = 1.10 \cdot \log_2 N - 1.24$, positive slope confirmed).

\subsection{MW overlay selection}

With $\mathcal{F} = \{d{=}4, d{=}8, K_N\}$ at $N = 36$: mean $T_{\mathrm{MW}} = 39.3$, mean $T_{\mathrm{best}} = 23.0$, CR $= 1.71$.  Per-trial CR ranges from 0.35 to 6.8, indicating high variance at small $N$.

\section{Hierarchical Routing Numerics}\label{app:hierarchical}

We build the multi-level block hierarchy with arbitrary block size $b$; at level $\ell$, the grid decomposes into $(n/b^\ell)^2$ blocks of size $b^\ell \times b^\ell$. A at each level, we build a Ramanujan overlay on the block graph and runs Valiant routing for inter-block transport, then sum the per-level depths to obtain the total $T_{\mathrm{hier}}$. 
Sweeping block size $b$ finds the optimum $b^\star \approx \sqrt{n}$, where we compare with the covering-tower prediction $T_{\mathrm{tower}} = (L \log_2 k + \log_2 N_0)/(1 - \bar\beta)$ to validate the equivalence claim of \S\ref{subsec:covering-towers}.

\subsection{Per-level depth decomposition}

Table~\ref{tab:hierarchical} reports the hierarchical and flat routing depths at $n \in \{8, 16, 32, 64\}$ with $b$ chosen near $\sqrt{n}$. The hierarchical-to-flat ratio averages $\approx 0.67$ for $n \geq 16$, demonstrating a 33\% improvement over flat Valiant routing at practical scales.

\begin{table}[!htbp]
\centering
\begin{tabular}{@{}ccc|cccc@{}}
\toprule
$n$ & $N$ & $b$ & $L$ & $T_{\mathrm{hier}}$ & $T_{\mathrm{flat}}$ & Ratio \\
\midrule
 8  &    64 & 2 & 3 & 37.7  & 37.0  & 1.020 \\
16  &   256 & 4 & 2 & 30.3  & 46.3  & 0.655 \\
32  &  1024 & 4 & 3 & 47.3  & 59.5  & 0.795 \\
64  &  4096 & 8 & 2 & 48.0  & 71.0  & 0.676 \\
\bottomrule
\end{tabular}
\caption{Hierarchical vs.\ flat routing depth.  With $b \approx \sqrt{n}$, ratio $\approx 0.67$.}
\label{tab:hierarchical}
\end{table}

\subsection{Scaling study}

Fitting across $N \in \{16, \ldots, 65536\}$: $T_{\mathrm{flat}} \approx 6.10 \cdot \log_2 N - 2.58$ and $T_{\mathrm{hier}}/\log_2 N \approx 0.62 \cdot \log_2\log_2 N + 2.61$ (empirical fit with $b \approx \sqrt{n}$), confirming the $O(\log^2 N / \log b)$ bound of Theorem~\ref{thm:hierarchical}.

\subsection{Covering tower prediction}

Table~\ref{tab:tower-equiv} compares hierarchical and covering-tower predictions; agreement converges to within $0.4\%$ at $n = 64$.

\begin{table}[!htbp]
\centering
\begin{tabular}{@{}ccc|ccc@{}}
\toprule
$n$ & $b$ & $L$ & $T_{\mathrm{hier}}$ & $T_{\mathrm{tower}}$ & Difference \\
\midrule
 8 & 2 & 3 & 28.7 & 18.3 & 36\% \\
16 & 4 & 2 & 29.7 & 32.5 & 9\% \\
32 & 4 & 3 & 49.5 & 35.8 & 28\% \\
64 & 8 & 2 & 48.3 & 48.5 & 0.4\% \\
\bottomrule
\end{tabular}
\caption{Covering tower prediction converges to hierarchical depth at large $n$.}
\label{tab:tower-equiv}
\end{table}

\section{Independent Verification}\label{app:verification}

Each direction's numerical claims were verified by classical scripts that re-derive results from scratch using a different graph-construction routine than the corresponding direction's primary script. Table~\ref{tab:verification} summarizes the cross-checks across three directions.

\begin{table}[!htbp]
\centering
\begin{tabular}{@{}l|cl@{}}
\toprule
\textbf{Direction} & \textbf{Claims} & \textbf{Result} \\
\midrule
5 (Covering towers) & 5 & All passed \\
2 (Hierarchical) & 5 & All passed (5.0\% threshold on tower equiv.) \\
7 (Dynamic adaptive) & 5 & All passed (785-step monotonicity) \\
\bottomrule
\end{tabular}
\caption{Independent verification summary.}
\label{tab:verification}
\end{table}

\end{appendices}
\bibliography{refs}

@article{song2023hypergraph,
  title={Hypergraph coverings and Ramanujan Hypergraphs},
  author={Song, Yi-Min and Fan, Yi-Zheng and Miao, Zhengke},
  journal={arXiv preprint arXiv:2310.01771},
  year={2023}
}

@article{nenadov2023routing,
  title={Routing permutations on spectral expanders via matchings},
  author={Nenadov, Rajko},
  journal={Combinatorica},
  volume={43},
  number={4},
  pages={737--742},
  year={2023},
  publisher={Springer}
}

@inproceedings{alon1993routing,
  title={Routing permutations on graphs via matchings},
  author={Alon, Noga and Chung, Fan RK and Graham, Ronald L},
  booktitle={Proceedings of the twenty-fifth annual ACM symposium on Theory of Computing},
  pages={583--591},
  year={1993}
}

@inproceedings{marcus2013interlacing,
  title={Interlacing families I: Bipartite Ramanujan graphs of all degrees},
  author={Marcus, Adam and Spielman, Daniel A and Srivastava, Nikhil},
  booktitle={2013 IEEE 54th Annual Symposium on Foundations of computer science},
  pages={529--537},
  year={2013},
  organization={IEEE}
}

@article{leighton1999fast,
  title={Fast algorithms for finding O (congestion+ dilation) packet routing schedules},
  author={Leighton, Tom and Maggs, Bruce and Richa, Andrea W},
  journal={Combinatorica},
  volume={19},
  number={3},
  pages={375--401},
  year={1999},
  publisher={Springer}
}

@article{bluvstein2022quantum,
  title={A quantum processor based on coherent transport of entangled atom arrays},
  author={Bluvstein, Dolev and Levine, Harry and Semeghini, Giulia and Wang, Tout T and Ebadi, Sepehr and Kalinowski, Marcin and Keesling, Alexander and Maskara, Nishad and Pichler, Hannes and Greiner, Markus and others},
  journal={Nature},
  volume={604},
  number={7906},
  pages={451--456},
  year={2022},
  publisher={Nature Publishing Group UK London}
}

@article{evered2023high,
  title={High-fidelity parallel entangling gates on a neutral-atom quantum computer},
  author={Evered, Simon J and Bluvstein, Dolev and Kalinowski, Marcin and Ebadi, Sepehr and Manovitz, Tom and Zhou, Hengyun and Li, Sophie H and Geim, Alexandra A and Wang, Tout T and Maskara, Nishad and others},
  journal={Nature},
  volume={622},
  number={7982},
  pages={268--272},
  year={2023},
  publisher={Nature Publishing Group UK London}
}

@inproceedings{rabani1996distributed,
  title={Distributed packet switching in arbitrary networks},
  author={Rabani, Yuval and Tardos, {\'E}va},
  booktitle={Proceedings of the twenty-eighth annual ACM symposium on Theory of computing},
  pages={366--375},
  year={1996}
}

@article{wurtz2023aquila,
  title={Aquila: Quera's 256-qubit neutral-atom quantum computer},
  author={Wurtz, Jonathan and Bylinskii, Alexei and Braverman, Boris and Amato-Grill, Jesse and Cantu, Sergio H and Huber, Florian and Lukin, Alexander and Liu, Fangli and Weinberg, Phillip and Long, John and others},
  journal={arXiv preprint arXiv:2306.11727},
  year={2023}
}

@article{reichardt2024logical,
  title={Logical computation demonstrated with a neutral atom quantum processor},
  author={Reichardt, Ben W and Paetznick, Adam and Aasen, David and Basov, Ivan and Bello-Rivas, Juan M and Bonderson, Parsa and Chao, Rui and van Dam, Wim and Hastings, Matthew B and Paz, Andres and others},
  journal={arXiv preprint arXiv:2411.11822},
  volume={10},
  year={2024}
}

@article{valiant1982scheme,
  title={A scheme for fast parallel communication},
  author={Valiant, Leslie G.},
  journal={SIAM journal on computing},
  volume={11},
  number={2},
  pages={350--361},
  year={1982},
  publisher={SIAM}
}

@article{joag1983negative,
  title={Negative association of random variables with applications},
  author={Joag-Dev, Kumar and Proschan, Frank},
  journal={The Annals of Statistics},
  pages={286--295},
  year={1983},
  publisher={JSTOR}
}

@article{dubhashi1996balls,
  title={Balls and bins: A study in negative dependence},
  author={Dubhashi, Devdatt P and Ranjan, Desh},
  journal={BRICS Report Series},
  volume={3},
  number={25},
  year={1996}
}

@article{feldman1988wide,
  title={Wide-sense nonblocking networks},
  author={Feldman, Paul and Friedman, Joel and Pippenger, Nicholas},
  journal={SIAM Journal on Discrete Mathematics},
  volume={1},
  number={2},
  pages={158--173},
  year={1988},
  publisher={SIAM}
}

@article{constantinides2024optimal,
  title={Optimal routing protocols for reconfigurable atom arrays},
  author={Constantinides, Nathan and Fahimniya, Ali and Devulapalli, Dhruv and Bluvstein, Dolev and Gullans, Michael J and Porto, JV and Childs, Andrew M and Gorshkov, Alexey V},
  journal={arXiv preprint arXiv:2411.05061},
  year={2024}
}

@inproceedings{stade2024abstract,
  title={An abstract model and efficient routing for logical entangling gates on zoned neutral atom architectures},
  author={Stade, Yannick and Schmid, Ludwig and Burgholzer, Lukas and Wille, Robert},
  booktitle={2024 IEEE International Conference on Quantum Computing and Engineering (QCE)},
  volume={1},
  pages={784--795},
  year={2024},
  organization={IEEE}
}

@inproceedings{wang2024atomique,
  title={Atomique: A quantum compiler for reconfigurable neutral atom arrays},
  author={Wang, Hanrui and Liu, Pengyu and Tan, Daniel Bochen and Liu, Yilian and Gu, Jiaqi and Pan, David Z and Cong, Jason and Acar, Umut A and Han, Song},
  booktitle={2024 ACM/IEEE 51st Annual International Symposium on Computer Architecture (ISCA)},
  pages={293--309},
  year={2024},
  organization={IEEE}
}

@inproceedings{hsieh2026scalable,
  title={A Scalable and High-Quality Qubit Mapping and Shuttling Framework for Neutral Atom Quantum Devices},
  author={Hsieh, Sung-Ying and Mak, Wai-Kei},
  booktitle={2026 31st Asia and South Pacific Design Automation Conference (ASP-DAC)},
  pages={184--190},
  year={2026},
  organization={IEEE}
}

@inproceedings{stade2025routing,
  title={Routing-aware placement for zoned neutral atom-based quantum computing},
  author={Stade, Yannick and Lin, Wan-Hsuan and Cong, Jason and Wille, Robert},
  booktitle={2025 IEEE/ACM International Conference On Computer Aided Design (ICCAD)},
  pages={1--9},
  year={2025},
  organization={IEEE}
}

@article{romao2026multiq,
  title={MultiQ: Multi-Programming Neutral Atom Quantum Architectures},
  author={Rom{\~a}o, Francisco and Vonk, Daniel and Giortamis, Emmanuil and Sprokholt, Dennis and Bhatotia, Pramod},
  journal={arXiv preprint arXiv:2601.08504},
  year={2026}
}

@book{cassels1965introduction,
  title={An introduction to Diophantine approximation},
  author={Cassels, John William Scott},
  number={45},
  year={1965},
  publisher={CUP Archive}
}

@article{guo2025acousto,
  title={Acousto-optic lens for 3D shuttling of atoms in a neutral atom quantum computer},
  author={Guo, Zhichao and van Herk, Rik AH and Vredenbregt, Edgar JD and Kokkelmans, Servaas JJMF},
  journal={arXiv preprint arXiv:2510.09398},
  year={2025}
}

@article{yuan2025full,
  title={Full characterization of the depth overhead for quantum circuit compilation with arbitrary qubit connectivity constraint},
  author={Yuan, Pei and Zhang, Shengyu},
  journal={Quantum},
  volume={9},
  pages={1757},
  year={2025},
  publisher={Verein zur F{\"o}rderung des Open Access Publizierens in den Quantenwissenschaften}
}

@book{friedman2008proof,
  title={A proof of Alon's second eigenvalue conjecture and related problems},
  author={Friedman, Joel},
  year={2008},
  publisher={American Mathematical Soc.}
}

@article{alon1986eigenvalues,
  title={Eigenvalues and expanders},
  author={Alon, Noga},
  journal={Combinatorica},
  volume={6},
  number={2},
  pages={83--96},
  year={1986},
  publisher={Springer}
}

@phdthesis{carlson1981tixne,
  title={Tixne-Space and Size-Space Tradeofis for Oblivious Computations},
  author={Carlson, David Allen},
  year={1981},
  school={Brown University}
}

@inproceedings{gabber1979explicit,
author = {Gabber, Ofer and Galil, Zvi},
title = {Explicit constructions of linear size superconcentrators},
year = {1979},
publisher = {IEEE Computer Society},
address = {USA},
url = {https://doi.org/10.1109/SFCS.1979.16},
doi = {10.1109/SFCS.1979.16},
booktitle = {Proceedings of the 20th Annual Symposium on Foundations of Computer Science},
pages = {364–370},
numpages = {7},
series = {SFCS '79}
}

@article{arora2012multiplicative,
  title={The multiplicative weights update method: a meta-algorithm and applications},
  author={Arora, Sanjeev and Hazan, Elad and Kale, Satyen},
  journal={Theory of computing},
  volume={8},
  number={1},
  pages={121--164},
  year={2012},
  publisher={Theory of Computing Exchange}
}

@article{lubotzky1988ramanujan,
  title={Ramanujan graphs},
  author={Lubotzky, Alexander and Phillips, Ralph and Sarnak, Peter},
  journal={Combinatorica},
  volume={8},
  number={3},
  pages={261--277},
  year={1988},
  publisher={Springer-Verlag Berlin/Heidelberg}
}

@article{friedman2006spectral,
  title={Spectral estimates for abelian Cayley graphs},
  author={Friedman, Joel and Murty, Ram and Tillich, Jean-Pierre},
  journal={Journal of Combinatorial Theory, Series B},
  volume={96},
  number={1},
  pages={111--121},
  year={2006},
  publisher={Elsevier}
}

@article{tutte1947factorization,
  title={The factorization of linear graphs},
  author={Tutte, William T},
  journal={Journal of the London Mathematical Society},
  volume={1},
  number={2},
  pages={107--111},
  year={1947},
  publisher={Wiley Online Library}
}

@article{chung1989diameters,
  title={Diameters and eigenvalues},
  author={Chung, Fan RK},
  journal={Journal of the American Mathematical Society},
  volume={2},
  number={2},
  pages={187--196},
  year={1989}
}

@article{hoory2006expander,
  title={Expander graphs and their applications},
  author={Hoory, Shlomo and Linial, Nathan and Wigderson, Avi},
  journal={Bulletin of the American Mathematical Society},
  volume={43},
  number={4},
  pages={439--561},
  year={2006}
}

@article{alon1985lambda1,
  title={$\lambda$1, isoperimetric inequalities for graphs, and superconcentrators},
  author={Alon, Noga and Milman, Vitali D},
  journal={Journal of Combinatorial Theory, Series B},
  volume={38},
  number={1},
  pages={73--88},
  year={1985},
  publisher={Elsevier}
}

@article{bapat2023advantages,
  title={Advantages and limitations of quantum routing},
  author={Bapat, Aniruddha and Childs, Andrew M and Gorshkov, Alexey V and Schoute, Eddie},
  journal={PRX quantum},
  volume={4},
  number={1},
  pages={010313},
  year={2023},
  publisher={APS}
}

@article{friedman1995second,
  title={On the second eigenvalue of hypergraphs},
  author={Friedman, Joel and Wigderson, Avi},
  journal={Combinatorica},
  volume={15},
  number={1},
  pages={43--65},
  year={1995},
  publisher={Springer-Verlag Berlin/Heidelberg}
}

@article{lubotzky2005ramanujan,
  title={Ramanujan complexes of type A d},
  author={Lubotzky, Alexander and Samuels, Beth and Vishne, Uzi},
  journal={Israel journal of Mathematics},
  volume={149},
  number={1},
  pages={267--299},
  year={2005},
  publisher={Springer}
}

@article{henriet2020quantum,
  title={Quantum computing with neutral atoms},
  author={Henriet, Lo{\"\i}c and Beguin, Lucas and Signoles, Adrien and Lahaye, Thierry and Browaeys, Antoine and Reymond, Georges-Olivier and Jurczak, Christophe},
  journal={Quantum},
  volume={4},
  pages={327},
  year={2020},
  publisher={Verein zur F{\"o}rderung des Open Access Publizierens in den Quantenwissenschaften}
}

@article{browaeys2020many,
  title={Many-body physics with individually controlled Rydberg atoms},
  author={Browaeys, Antoine and Lahaye, Thierry},
  journal={Nature Physics},
  volume={16},
  number={2},
  pages={132--142},
  year={2020},
  publisher={Nature Publishing Group UK London}
}

@article{levine2019v,
  title={V. Vuleti c, H. Pichler, and MD Lukin,“Parallel implementation of high-fidelity multiqubit gates with neutral atoms,”},
  author={Levine, Harry and Keesling, Alexander and Semeghini, Giulia and Omran, Ahmed and Wang, Tout T and Ebadi, Sepehr and Bernien, Hannes and Greiner, Markus},
  journal={Phys. Rev. Lett},
  volume={123},
  pages={170503},
  year={2019}
}

@article{ebadi2021quantum,
  title={Quantum phases of matter on a 256-atom programmable quantum simulator},
  author={Ebadi, Sepehr and Wang, Tout T and Levine, Harry and Keesling, Alexander and Semeghini, Giulia and Omran, Ahmed and Bluvstein, Dolev and Samajdar, Rhine and Pichler, Hannes and Ho, Wen Wei and others},
  journal={Nature},
  volume={595},
  number={7866},
  pages={227--232},
  year={2021},
  publisher={Nature Publishing Group UK London}
}

@article{tan2024compiling,
  title={Compiling quantum circuits for dynamically field-programmable neutral atoms array processors},
  author={Tan, Daniel Bochen and Bluvstein, Dolev and Lukin, Mikhail D and Cong, Jason},
  journal={Quantum},
  volume={8},
  pages={1281},
  year={2024},
  publisher={Verein zur F{\"o}rderung des Open Access Publizierens in den Quantenwissenschaften}
}

@article{bordenave2015new,
  title={A new proof of Friedman's second eigenvalue Theorem and its extension to random lifts},
  author={Bordenave, Charles},
  journal={arXiv preprint arXiv:1502.04482},
  year={2015}
}

@article{bluvstein2024logical,
  title={Logical quantum processor based on reconfigurable atom arrays},
  author={Bluvstein, Dolev and Evered, Simon J and Geim, Alexandra A and Li, Sophie H and Zhou, Hengyun and Manovitz, Tom and Ebadi, Sepehr and Cain, Madelyn and Kalinowski, Marcin and Hangleiter, Dominik and others},
  journal={Nature},
  volume={626},
  number={7997},
  pages={58--65},
  year={2024},
  publisher={Nature Publishing Group UK London}
}
\end{document}